\PassOptionsToPackage{unicode}{hyperref}
\PassOptionsToPackage{hyphens}{url}
\PassOptionsToPackage{dvipsnames,svgnames,x11names}{xcolor}
\documentclass[
  12pt]{article}

\usepackage{amsmath,amssymb}
\usepackage{amsthm}
\usepackage{array}
\usepackage{subcaption}
\usepackage{changepage}
\usepackage{makecell}
\usepackage{rotating}
\newcommand{\mathlarger}[1]{\mbox{\large$#1$}}
\usepackage{multirow}
\newtheorem{theorem}{Theorem}

\usepackage{iftex}
\ifPDFTeX
  \usepackage[T1]{fontenc}
  \usepackage[utf8]{inputenc}
  \usepackage{textcomp} 
\else 
  \usepackage{unicode-math}
  \defaultfontfeatures{Scale=MatchLowercase}
  \defaultfontfeatures[\rmfamily]{Ligatures=TeX,Scale=1}
\fi
\usepackage{lmodern}
\ifPDFTeX\else  
\fi
\IfFileExists{upquote.sty}{\usepackage{upquote}}{}
\IfFileExists{microtype.sty}{
  \usepackage[]{microtype}
  \UseMicrotypeSet[protrusion]{basicmath} 
}{}
\makeatletter
\@ifundefined{KOMAClassName}{
  \IfFileExists{parskip.sty}{%
    \usepackage{parskip}
  }{
    \setlength{\parindent}{0pt}
    \setlength{\parskip}{6pt plus 2pt minus 1pt}}
}{
  \KOMAoptions{parskip=half}}
\makeatother
\usepackage{xcolor}
\makeatletter
\ifx\paragraph\undefined\else
  \let\oldparagraph\paragraph
  \renewcommand{\paragraph}{
    \@ifstar
      \xxxParagraphStar
      \xxxParagraphNoStar
  }
  \newcommand{\xxxParagraphStar}[1]{\oldparagraph*{#1}\mbox{}}
  \newcommand{\xxxParagraphNoStar}[1]{\oldparagraph{#1}\mbox{}}
\fi
\ifx\subparagraph\undefined\else
  \let\oldsubparagraph\subparagraph
  \renewcommand{\subparagraph}{
    \@ifstar
      \xxxSubParagraphStar
      \xxxSubParagraphNoStar
  }
  \newcommand{\xxxSubParagraphStar}[1]{\oldsubparagraph*{#1}\mbox{}}
  \newcommand{\xxxSubParagraphNoStar}[1]{\oldsubparagraph{#1}\mbox{}}
\fi
\makeatother

\usepackage{longtable,booktabs,array}
\usepackage{calc} 
\usepackage{etoolbox}
\makeatletter
\patchcmd\longtable{\par}{\if@noskipsec\mbox{}\fi\par}{}{}
\makeatother
\IfFileExists{footnotehyper.sty}{\usepackage{footnotehyper}}{\usepackage{footnote}}
\makesavenoteenv{longtable}
\usepackage{graphicx}
\makeatletter
\def\maxwidth{\ifdim\Gin@nat@width>\linewidth\linewidth\else\Gin@nat@width\fi}
\def\maxheight{\ifdim\Gin@nat@height>\textheight\textheight\else\Gin@nat@height\fi}
\makeatother
\setkeys{Gin}{width=\maxwidth,height=\maxheight,keepaspectratio}
\makeatletter
\def\fps@figure{htbp}
\makeatother

\makeatletter
\@ifpackageloaded{caption}{}{\usepackage{caption}}
\AtBeginDocument{%
\ifdefined\contentsname
  \renewcommand*\contentsname{Table of contents}
\else
  \newcommand\contentsname{Table of contents}
\fi
\ifdefined\listfigurename
  \renewcommand*\listfigurename{List of Figures}
\else
  \newcommand\listfigurename{List of Figures}
\fi
\ifdefined\listtablename
  \renewcommand*\listtablename{List of Tables}
\else
  \newcommand\listtablename{List of Tables}
\fi
\ifdefined\figurename
  \renewcommand*\figurename{Figure}
\else
  \newcommand\figurename{Figure}
\fi
\ifdefined\tablename
  \renewcommand*\tablename{Table}
\else
  \newcommand\tablename{Table}
\fi
}
\@ifpackageloaded{float}{}{\usepackage{float}}
\floatstyle{ruled}
\@ifundefined{c@chapter}{\newfloat{codelisting}{h}{lop}}{\newfloat{codelisting}{h}{lop}[chapter]}
\floatname{codelisting}{Listing}

\makeatother
\makeatletter
\@ifpackageloaded{caption}{}{\usepackage{caption}}
\@ifpackageloaded{subcaption}{}{\usepackage{subcaption}}
\makeatother

\ifLuaTeX
  \usepackage{selnolig}  
\fi
\usepackage[]{natbib}
\usepackage{bookmark}

\IfFileExists{xurl.sty}{\usepackage{xurl}}{} 
\hypersetup{
  pdftitle={Title},
  pdfauthor={Author 1; Author 2},
  pdfkeywords={3 to 6 keywords, that do not appear in the title},
  colorlinks=true,
  linkcolor={blue},
  filecolor={Maroon},
  citecolor={Blue},
  urlcolor={Blue},
  pdfcreator={LaTeX via pandoc}}

\newcommand{\anon}{1}

\begin{document}

\def\spacingset#1{\renewcommand{\baselinestretch}%
{#1}\small\normalsize} \spacingset{1}


\if1\anon
{
  \title{\bf Factor State Space Modelling of the Ornstein-Uhlenbeck Process with Measurement Error and its Application}
  \author{Shanglun Li, Toby Kenney\thanks{
    Corresponding author Toby Kenney, Department of Mathematics and Statistics, Dalhousie University, Halifax, NS, Canada. tkenney@mathstat.dal.ca}\hspace{.2cm}, and Hong Gu\\
    Department of Mathematics and Statistics, Dalhousie University}
  \maketitle
} \fi

\if0\anon
{
  \bigskip
  \bigskip
  \bigskip
  \begin{center}
    {\LARGE\bf Title}
\end{center}
  \medskip
} \fi

\bigskip
\begin{abstract}
Standard Ornstein-Uhlenbeck (OU) models often yield biased parameter estimates when measurement error is ignored. While the Ornstein-Uhlenbeck State Space Model (OUSSM) addresses this in univariate settings, multidimensional extensions remain limited. This paper introduces the factor OUSSM to model multi-dimensional, mean-reverting systems with observational noise. We resolve critical identifiability challenges in parameter estimation by establishing necessary constraints and validating the method through extensive simulations. We demonstrate the model's versatility by analyzing human gut microbiome dynamics and North Atlantic Sea Surface Temperature (SST) data. The results reveal distinct latent temporal structures in both biological and environmental systems, establishing the factor OUSSM as a robust framework for multivariate time series analysis.
\end{abstract}

\noindent%
\vfill

\newpage
\spacingset{1.8} 

\section{Introduction}
\label{sec1}
The Ornstein-Uhlenbeck (OU) process is widely used to model
mean-reverting processes with numerous real-world applications. For
instance, \citet{WANG2015363} and \citet{liang2011optimal} applied the
OU process in the financial domain, while \citet{yatabe2021ornstein}
used it to model human body weight fluctuations. Additionally,
\citet{kenney2020application} and \citet{wang2018stochastic} explored
its applications in human microbiome and epidemiological studies,
respectively. However, \citet{kenney2020application} noted that an OU
process not accounting for measurement error can lead to bias in the
estimation of parameters when such errors are present in the data. To
address this issue, a state space model can be used to incorporate
measurement error into the data analysis.

The OU process also shares theoretical connections with autoregressive
(AR) and autoregressive moving average (ARMA) models. The OU process
can be seen as a continuous-time counterpart to AR models, with the
mean-reversion rate corresponding to the decay parameter in AR
models. When incorporating measurement error, the equally-spaced
discrete-time sampling of an OU model with measurement error follows
an ARMA process, with the uncertainty over the true latent value of
the process giving rise to the moving average terms in the ARMA
model. These links underline the relevance of Ornstein-Uhlenbeck State
Space Models (OUSSMs) as a robust tool for handling noisy,
multivariate systems where traditional AR or ARMA models may fall
short. OUSSMs thus provide a flexible framework for analyzing complex
real-world processes with both latent dynamics and measurement
uncertainty.

The OUSSM, where the OU process is used to model the state equation, has been used in a variety of applications, including population abundance modeling~\citep{dennis2014density}, sea surface temperature prediction~\citep{tandeo2011linear}, and animal location tracking~\citep{johnson2008continuous}. In population ecology, \cite{dennis2014density} modeled the temporal dynamics of single-species abundance as a one-dimensional OU process, where the latent population size evolves continuously under mean reversion toward an equilibrium and the observation equation accounts for sampling noise. \cite{tandeo2011linear} extended the linear Gaussian state-space framework to analyze sea-surface temperature (SST) anomalies at a fixed spatial location, assuming a univariate continuous-time OU process for the latent temperature field and explicitly accommodating irregular satellite sampling intervals and measurement error through Kalman filtering and an EM algorithm. Similarly, \cite{johnson2008continuous} developed a continuous-time correlated random-walk model for animal telemetry data by integrating a univariate OU velocity process; the resulting model captures directional persistence in movement while estimating location uncertainty within a state-space framework. Despite their methodological differences, these models share a common one-dimensional structure--each tracks a single latent process per time series without modeling cross-series dependencies or shared latent variability. 

Many real-world systems, such as ecological communities or microbiome dynamics, involve multiple interacting components that evolve jointly over time. Indeed, recent studies of the human microbiome highlight that taxa do not vary independently but form complex temporal interaction networks~\citep{lugo2019dynamic}, and that community stability depends on the strength and structure of these interactions~\citep{he2022arzimm}. Capturing such dependencies requires extending the OUSSM to a multidimensional framework capable of representing inter-taxa coupling and shared latent drivers.

Multivariate (vector) Ornstein--Uhlenbeck (OU) processes have been employed in several disciplines to model systems of correlated stochastic dynamics. In phylogenetic comparative biology, multivariate OU models describe the joint evolution of multiple traits across a tree, allowing correlated adaptation and shared evolutionary optima among traits~\citep{bartoszek2012phylogenetic}. Similarly, in neuroscience, multivariate OU formulations have been used to infer directed effective connectivity among large-scale brain regions from fMRI time series~\citep{insabato2018bayesian}. However, these models typically treat the observed values as noise-free realizations of the underlying process or assume negligible measurement error. As a result, they do not explicitly account for the observational uncertainty that is pervasive in biological and ecological data.

This paper addresses these limitations by developing the factor OUSSM, a multidimensional extension of the OUSSM inspired by dynamic factor models in state-space analysis~\citep{durbin2012time}. The term ``factor'' reflects that the latent process is represented by a small number of underlying factors--analogous to principal components in factor analysis--that drive the dynamics of multiple observed time series while incorporating measurement error in the observation equation. This structure enables the model to capture shared stochastic variation, cross-process dependence, and noise explicitly within a likelihood-based framework.

The definition and development of factor OUSSM is given in Section~\ref{sec2}.
Section~\ref{sec4} 
presents simulation studies to assess the performance of the proposed
model in a wide variety of cases. In Section~\ref{sec5}, we apply this model to both the human gut
microbiome data of \citet{caporaso2011moving} and to sea-surface
temperature (SST) data, demonstrating its applicability to distinct
real-world systems. Finally, Section~\ref{sec6}
concludes the paper and discusses potential directions for future work.

\section{Factor OUSSM}
\label{sec2}
The Ornstein-Uhlenbeck (OU) process is a stochastic differential
equation (SDE) model widely employed to model the dynamic behavior of
a continuous-time stochastic process. It satisfies the following
linear stochastic differential equation:
$$dx_t =-\theta(x_t -\mu)dt+\sigma dW_t$$
Here, $\{W_t: t \geq 0\}$ is
a Brownian motion term. $\mu$ represents the mean value towards which
the process tends to revert, $\theta$ represents the rate of mean
reversion, and $\sigma^2$ denotes the diffusion variance (See
\citet{finch2004ornstein}, for example). While the SDE can be solved for
$\theta < 0$, the parameter $\theta$ is usually constrained to be
positive, which ensures that the system has a stationary
state.

This stochastic differential equation is solvable, and the solution is
given by:
$$x_{t_{n+1}} = e^{-\theta \Delta t_{n+1}} x_{t_n} + (1-e^{-\theta\Delta
  t_{n+1}})\mu + \eta_{t_{n+1}}, \eta_{t_{n+1}} \sim N\left(0, Q_{t_{n+1}}\right)$$
where $\Delta t_{n+1} = t_{n+1}-t_n$ and $\eta_{t_{n+1}}$ represents a
Gaussian noise term with variance $Q_{t_{n+1}}=
\dfrac{\sigma^2}{2\theta}\left(1-e^{-2\theta\Delta t_{n+1}}\right)$. This
form makes it easier to use the OU process for modelling data sampled
at discrete timepoints. The stationary distribution is Gaussian with
mean $\mu$ and variance $\dfrac{\sigma^2}{2\theta}$.

The multidimensional Ornstein-Uhlenbeck (OU) process is a trivial
extension of the one-dimensional OU process into multiple
dimensions. This process is often used to model vector-valued data
with dependencies across dimensions and mean-reverting behavior in
each dimension. The continuous-time model for the multidimensional OU
process is governed by the following stochastic differential equation
(SDE):
$$dx_t = -\Theta(x_t-\mu)d_t+\Sigma^{1/2} dW_t$$ where $x_t$ is an
$m$-dimensional state vector at time t, $\Theta$ is an $m\times m$
mean-reversion matrix, $\mu$ is an $m$-dimensional mean vector to
which $x_t$ reverts, $\Sigma$ is a symmetric $m\times m$ diffusion matrix, and
$W_t$ is an $m$-dimensional Brownian motion. We see that this is the
same SDE as the OU process, but with parameters replaced by vectors
and matrices as appropriate. The constraint $\theta >0$ in the
1-dimensional case is replaced by the constraint that the eigenvalues
of the matrix $\Theta$ should have positive real part.

Given a time step $\Delta{t_n}$, the solution to the multidimensional OU
process is (See, for example, \citet{gardiner2009handbook}):
$$x_{t_{n+1}} = e^{-\Theta \Delta t_{n+1}} x_{t_n} + (I-e^{-\Theta\Delta
  t_{n+1}})\mu + \eta_{t_{n+1}}, \eta_{t_{n+1}} \sim N\left(0, Q_{t_{n+1}}\right)$$
where $e^{-\Theta \Delta t_{n+1}}$ is the matrix exponential of $-\Theta
\Delta t_{n+1}$, which determines how quickly the process reverts to the
mean $\mu$, $\eta_{t_{n+1}}$ is a normally distributed noise term with
mean zero and covariance matrix $Q_{t_{n+1}}$, where $Q_{t_{n+1}}$ is the
solution to $\Theta Q_{t_{n+1}} + Q_{t_{n+1}}\Theta^T = \Sigma - e^{-\Theta
  \Delta t_{n+1}}\Sigma e^{-\Theta^T \Delta t_{n+1}}$. The stationary distribution
is Gaussian with mean $\mu$ and covariance matrix $Q_{t_\infty}$
satisfying $\Theta Q_{t_\infty} + Q_{t_\infty}\Theta^T = \Sigma$.

The dynamics of this system depend on the eigenvalues of
$\Theta$. Assuming the eigenvalues are distinct, if they are real,
then $\Theta$ is diagonalizable over the reals, and changing to the
basis for which $\Theta$ is diagonal, we can view the system as a
number of OU processes with correlated diffusion. If $\Theta$ has
complex eigenvalues, then it is not diagonalizable over the reals. If
the eigenvalues are distinct, then the matrices are
block-diagonalizable into $2\times 2$ or $1\times 1$ blocks over the
reals. The $2\times 2$ blocks represent a decaying cyclic trend
between two systems, i.e. a spiral pattern, while the $1\times 1$
blocks represent separate decaying systems. The diffusion processes of
all these may be correlated. Figures~\ref{realpic} and~\ref{complex}
shows examples of the two types of dynamics in two dimensions.

\begin{figure}[h]
\centering
\begin{subfigure}{\textwidth}
\centering
\includegraphics[width = 0.4\textwidth]{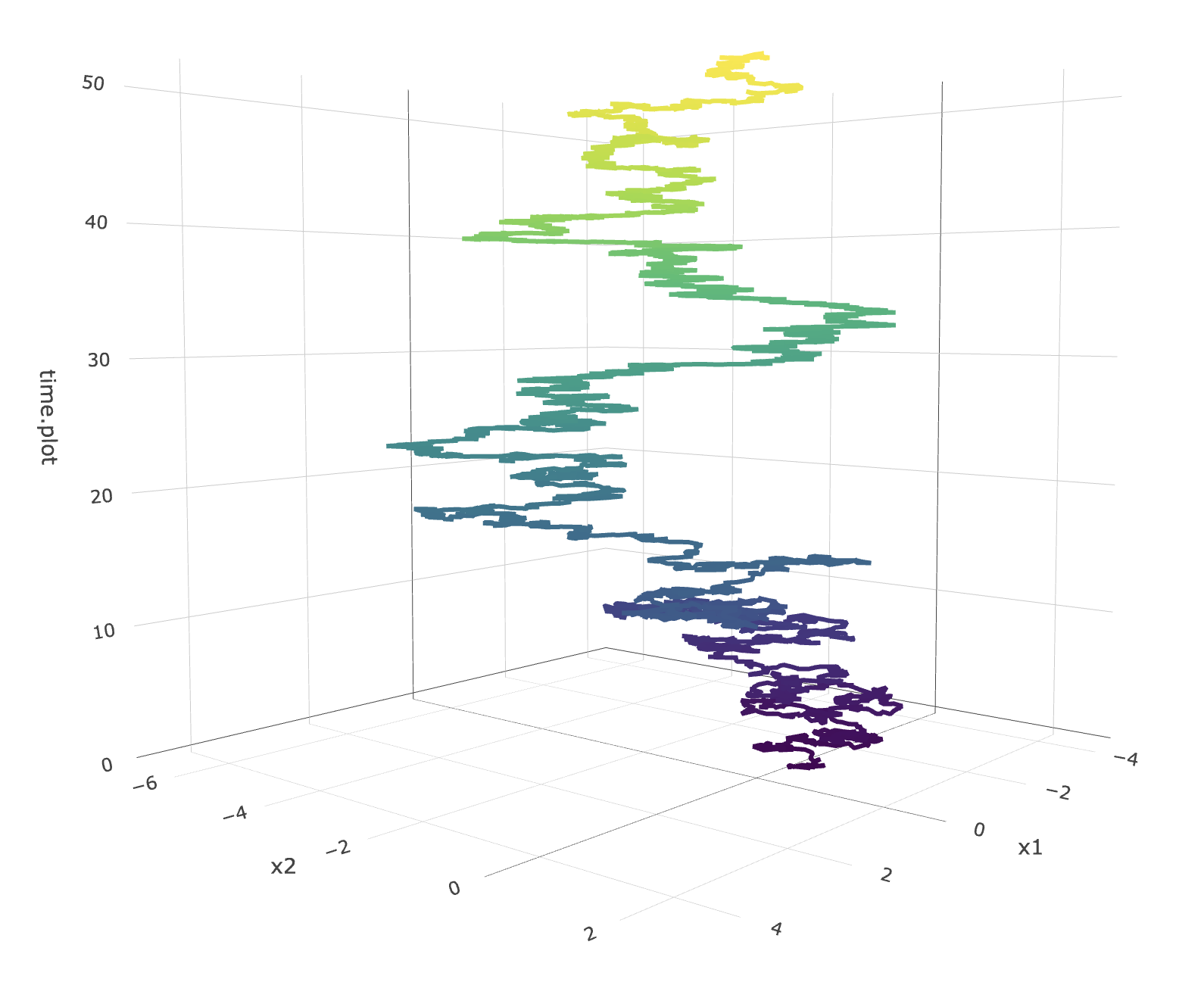	}
\includegraphics[width = 0.4\textwidth]{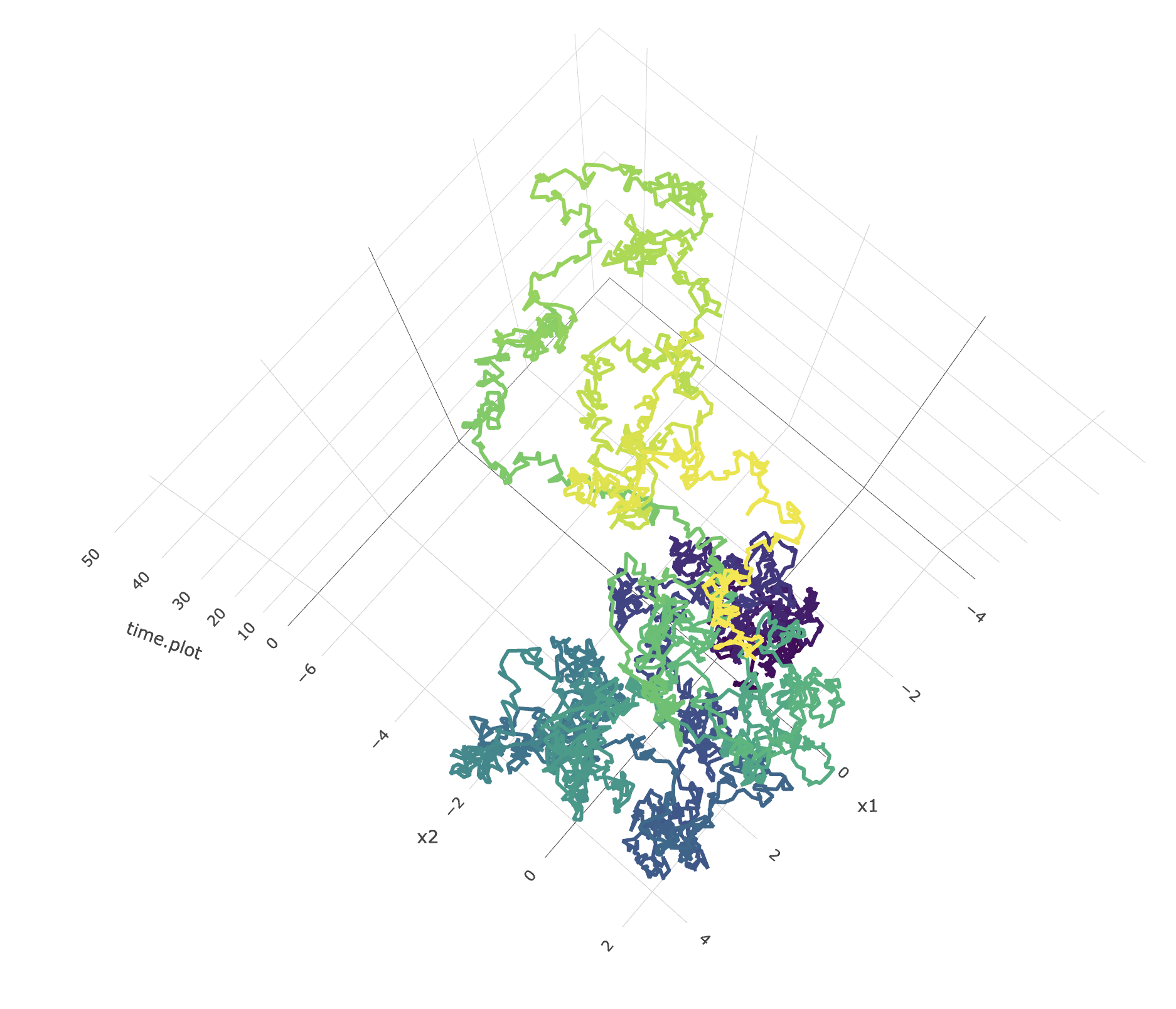	}
\caption{Real Eigenvalues}\label{realpic}

\end{subfigure}

\begin{subfigure}{\textwidth}
\centering
\includegraphics[width = 0.4\textwidth]{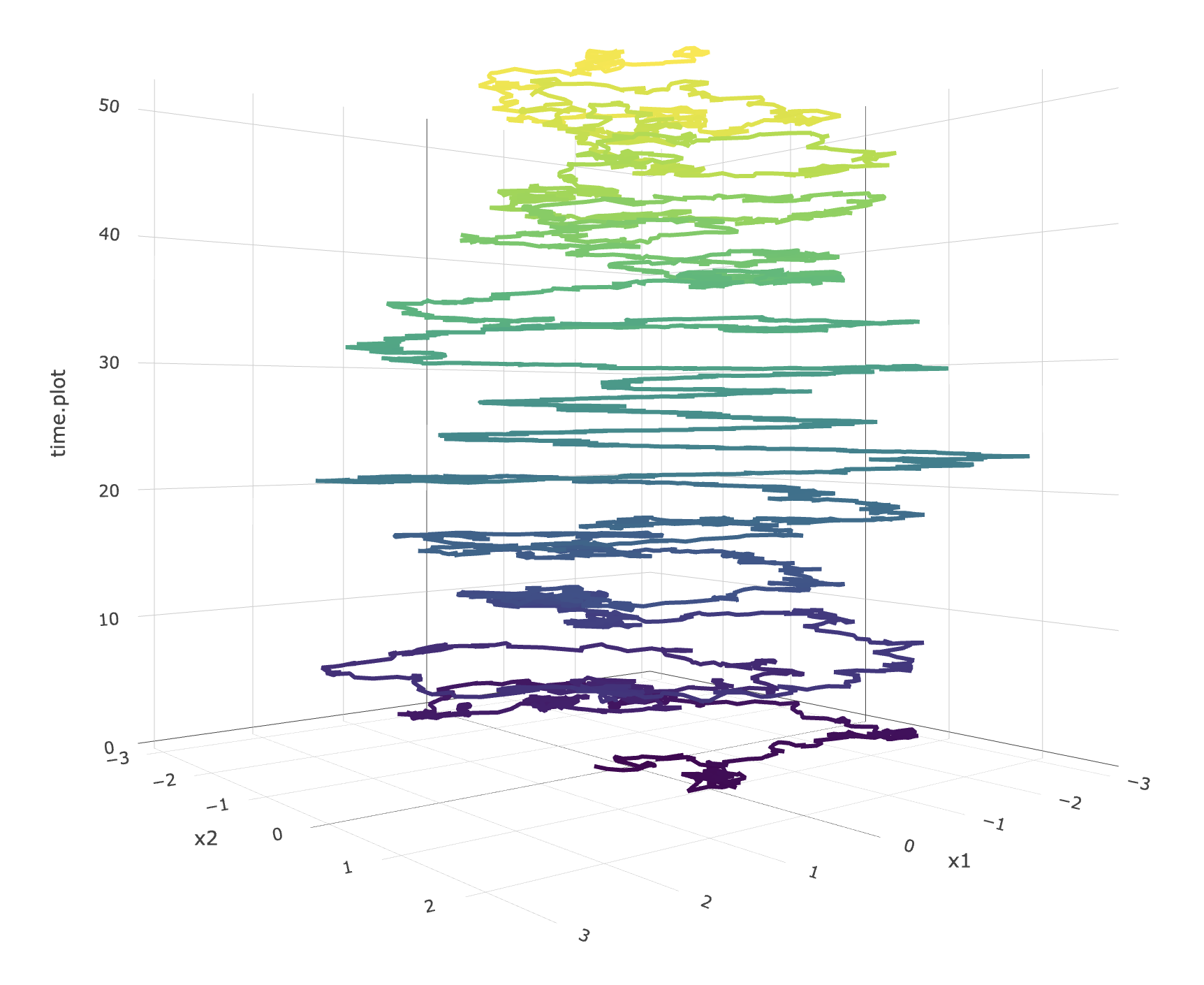	}
\includegraphics[width = 0.4\textwidth]{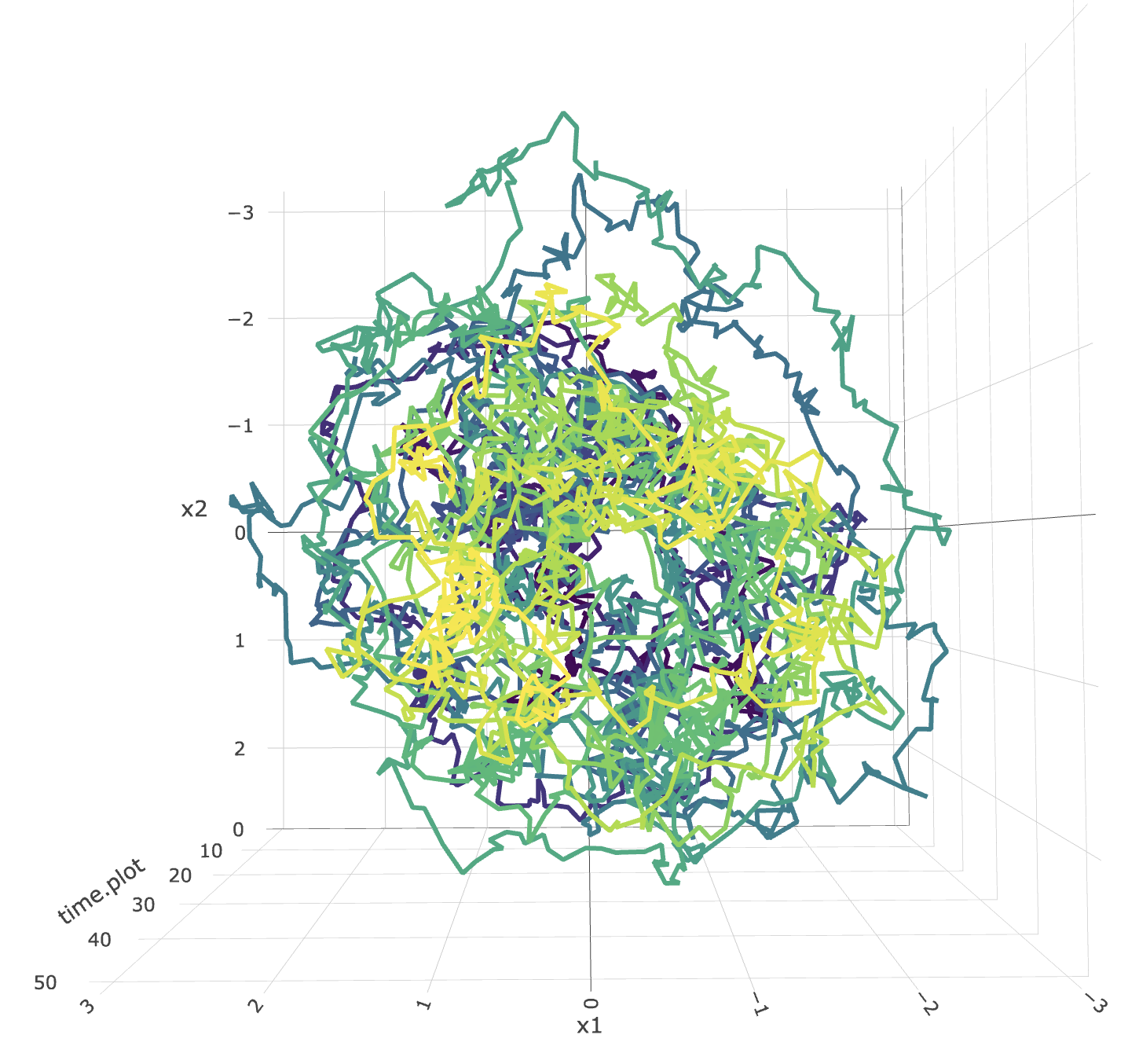	}
\caption{Complex eigenvalues}\label{complex}
\end{subfigure}

\caption{The OUSSM dynamics in two dimensional space when $\Theta$ has
  two real eigenvalues or two complex eigenvalues.}
\end{figure}

The factor Ornstein-Uhlenbeck State Space Model (OUSSM) is a
multidimensional state space model with a multidimensional OU process
as the state equation. It not only effectively handles the issue of
unaccounted measurement error in the OU process, but also accounts for
multidimensional interactions. The state is linked to the observed
data through a loading matrix. Solving the OU equation for the
discrete timepoints, we get the following:
\begin{align}
y_{t_{n}} &= \mu + Z x_{t_{n}} + \varepsilon, \varepsilon \sim N(0,H)\nonumber \\
x_{t_{n+1}} &=  e^{-\Theta \Delta t_{n+1}} x_{t_n} + \eta_{t_{n+1}}, \eta_{t_{n+1}} \sim N\left(0, Q_{t_{n+1}}\right) \tag{1} \label{eq1}
\end{align}
Within this model, the observed vector-valued variable at time $t_n$
is represented by $y_{t_n}$. The term $\varepsilon$ represents the
measurement error, assumed to follow a multivariate normal
distribution with mean zero and covariance matrix $H$. The underlying
vector-valued latent state variable is denoted by $x_{t_n}$; $\Delta
t_{n+1} = t_{n+1} - t_n$ is the time gap between two successive time
points. At the initial time point, $t_0$, the state vector $x_{t_0}$
follows a multivariate normal distribution with mean $a_{t_0}$ and
covariance matrix $P_{t_0}$. Notably, we have moved the mean value
towards which the process tends to revert ($\mu$) to the observation
equation, simplifying the system. Consequently, the state equation is
assumed to exhibit mean-reverting behavior around zero. The term
$\eta_{t_{n+1}}$ is the diffusion error term distributed as a multivariate
normal distribution with mean 0 and covariance matrix $Q_{t_{n+1}}$, where
$Q_{t_{n+1}}$ satisfies $\Theta Q_{t_{n+1}} + Q_{t_{n+1}}\Theta^T = \Sigma -
e^{-\Theta \Delta t_{n+1}}\Sigma e^{-\Theta^T \Delta
  t_{n+1}}$. Table~\ref{parameter} shows the dimension of each variable
where $m$ is the dimension of the latent OU process and $p$ is the
dimension of the observed data.

\begin{table}[h]
\centering
\caption{Parameter dimensions of factor OUSSM model}
\label{parameter}
\begin{tabular}{ |c|c|c|c|  }
 \hline
 \multicolumn{2}{|c|}{Vector} &  \multicolumn{2}{|c|}{Matrix}\\\hline
$y_{t_n}$ & $p \times 1$ & $Z$ & $p \times m$\\ 
$x_{t_n}$ & $m \times 1$ & $\Theta$ & $m \times m$\\ 
$\mu$ & $p \times 1$ & $\Sigma$ & $m \times m$\\ 
$\varepsilon_{t_n}$ & $p \times 1$ & $H$ & $p \times p$\\ 
$\eta_{t_n}$ & $m \times 1$ & $Q_{t_n}$ & $m \times m$ \\ 
$a_{t_0}$ & $m \times 1$ & $P_{t_0}$ & $m \times m$\\ \hline
\end{tabular}
\end{table}

In the factor OUSSM model, because the latent state can be expressed
in an arbitrary basis, if no constraints are imposed on the parameters
during estimation, identifiability issues arise. To prevent this, it
is necessary to introduce appropriate constraints on the parameters.

\subsection{Model Identifiability}
\label{sec3}
In the factor OUSSM model, for an invertible matrix $A$, if we let $\tilde x = Ax$ and $\tilde Z = ZA^{-1}$, it is easy to see the following output equations are equivalent: 
\begin{eqnarray*}
y_{t_{n}} &=& \mu + Zx_{t_{n}} + \varepsilon_{t_n}\\
 &=& \mu + ZA^{-1} A x_{t_{n}} + \varepsilon_{t_n}\\
 &=&  \mu + \tilde Z \tilde x_{t_{n}} + \varepsilon_{t_n}
\end{eqnarray*}
Then, if we look at the state equation:
\begin{align}
d\tilde x &= Adx = -A\Theta A^{-1}A x dt + A\Sigma^{1/2}dW_t \nonumber \\
&= -A\Theta A^{-1}\tilde x dt + A\Sigma^{1/2}dW_t \tag{2} \label{eq2}
\end{align}
Therefore, any choice of invertible $A$ will not change the model,
which causes an identifiability problem. We want to add constraints on
$\Theta$ and $\Sigma$ so that we have unique solutions and at the same
time the problem is simplified.

\begin{theorem}\label{identifiability_theorem}
  For any real matrices $\Theta$ and $\Sigma$, where $\Sigma$ is
  symmetric and positive definite, there is an invertible real matrix
  $A$, such that $A\Sigma A^T = I$, and $\tilde \Theta = A \Theta
  A^{-1}$ is a sum of an antisymmetric matrix plus a diagonal matrix
  with decreasing diagonal elements. Furthermore, if the diagonal
  elements of $\tilde \Theta$ are distinct, then $\tilde \Theta$ is
  unique up to conjugation by a diagonal matrix whose square is the
  identity matrix. In this case, the matrix $A$ is also unique up to
  left multiplication by a diagonal matrix with square the identity.
\end{theorem}
  
The proof of Theorem~\ref{identifiability_theorem} is in
Appendix~A. Based on this theorem, to ensure there are no
identifiability issues, we impose the constraint that $\Theta$ be the
sum of an antisymmetric matrix and a diagonal matrix with decreasing
diagonal elements and that the elements in the first row of $Z$ all be
non-negative. Additionally, we require $\Sigma$ to be the identity
matrix. In this paper, we assume that the measurement errors in each
dimension are independent, thus $H$ is modelled as a diagonal
matrix. 

In summary, under the constraints, the updated state equation is
\begin{align*}
y_{t_{n}} &= \mu + Z x_{t_{n}} + \varepsilon, \varepsilon \sim N(0,H) \\
x_{t_{n+1}} &=  e^{-\Theta \Delta t_{n+1}} x_{t_n} + \eta_{t_{n+1}}, \eta_{t_{n+1}} \sim N\left(0, Q_{t_{n+1}}\right) 
\end{align*}
Here, $Z$ is the loading matrix with $Z_{1j} > 0$ for
$j=1,\ldots,m$. The covariance matrix of the measurement error, $H$,
is diagonal. The covariance matrix of the diffusion error,
$\eta_{t_{n+1}}$, satisfies $\Theta Q_{t_{n+1}} + Q_{t_{n+1}}\Theta^T = I -
e^{-\Theta \Delta t_{n+1}} e^{-\Theta^T \Delta t_{n+1}}$. The matrix $\Theta$ is a
sum of an antisymmetric matrix plus a diagonal matrix with decreasing
diagonal elements.

\subsection{Parameter Interpretation}
\label{para_interpretation}

The factor OUSSM involves several parameter matrices---$\Theta$, $\Sigma$, $Z$, and $H$---each of which plays a distinct role in representing the latent dynamics and their relationship to the observed data. To facilitate interpretation and clarify the identifiability constraints, we briefly summarize the meaning of these parameters and the rationale for their specification, using microbial dynamics as an example.

\paragraph*{State representation} 
Here ${x}_t$ represent latent processes that drive the dynamics of multiple observed microbial taxa. The matrix $Z$ links these latent factors to the observed abundances through the observation equation
\[
{y}_t = {\mu} + Z {x}_t + {\varepsilon}_t, \qquad 
{\varepsilon}_t \sim {N}(0,H),
\]
where ${y}_t$ is $p$-dimensional, ${x}_t$ is $m$-dimensional, and $Z$ is a $p\times m$ loading matrix. Each column of $Z$ can therefore be viewed as a temporal factor loading that quantifies how strongly each taxon is influenced by the corresponding latent mean-reverting component of ${x}_t$. Increasing the dimension $m$ allows the model to represent more complex latent dynamics at the cost of additional parameters.

From a mathematical perspective, the latent state ${x}_t$ summarizes shared temporal patterns in the observed data, capturing mean-reverting fluctuations that underlie multiple taxa. Whether these latent components have a direct biological interpretation depends on the context of the study. In some cases, each latent dimension may correspond to an underlying physiological or ecological process--such as digestive function, immune activity, or environmental exposure--that jointly influences several microbial taxa. In other cases, the latent factors simply summarize correlated temporal variation among taxa, without implying specific biological mechanisms. Thus, the latent state can be interpreted either as a mathematical summary of shared dynamics or, when supported by external evidence, as a proxy for unobserved biological systems governing microbiome behavior.

\paragraph*{Role of $Z$ and model identifiability}
It is worth noting that, unlike a usual factor analysis, it is possible for the latent state to have higher dimension than the observed state. For example, a one-dimensional process could have latent mean reversion on multiple time-scales, represented in this model by $p=1,m>1$.

If $Z=I$ (when $m=p$), the latent states coincide directly with the observed trajectories, and the latent dynamics can be interpreted as the underlying dynamics for each observed variable. This formulation is simple but precludes dimension reduction, since a large number of observed time series cannot be represented by a smaller set of shared latent factors. By contrast, allowing a general $Z$ enables each observed series to be driven by multiple latent OU processes, providing a flexible low-rank representation that can capture correlated fluctuations across taxa. In this framework, $\Theta$ captures intrinsic mean-reversion structure, while $Z$ governs how these latent factors manifest in each observed taxon. This formulation yields clearer interpretability and greater flexibility for modeling complex multivariate microbiome dynamics, particularly when $\Theta$ is diagonalized.

\paragraph*{Diagonalization of $\Theta$}
The mean-reversion matrix $\Theta$ determines the temporal dependence structure among latent components. When $\Theta$ has distinct real eigenvalues, it can be diagonalized, and the latent processes correspond to multiple mean-reverting processes with rates given by the eigenvalues with correlated diffusion errors. When $\Theta$ has complex conjugate eigenvalues, it can be block-diagonalized into $2\times2$ or $1\times1$ real blocks, corresponding to oscillatory or cyclic behavior between pairs of latent components. For interpretability, we present results in the transformed basis where $\Theta$ is diagonal (or block-diagonal), so that each latent factor represents a dynamical component with its own characteristic timescale. This diagonalized representation simplifies the interpretation of microbial dynamics in the data application, where, for example, one latent component may correspond to a fast-fluctuating taxon group and another to a slow-reverting trend.

\paragraph*{Summary of interpretation}
The imposed constraints ensure both identifiability and interpretability: $Z$ quantifies how the latent OU factors contribute to each observed variable, $\Theta$ encodes the temporal coupling and mean-reversion strength among latent components, $\Sigma$ defines the diffusion scale of the latent process, and $H$ captures independent measurement errors.

\subsection{Parameter Estimation}
\label{estimate}
The most commonly employed method for parameter estimation in this context is the Kalman filter, a recursive algorithm that allows for
the iterative calculation of the likelihood (See, for example, \citet{durbin2012time}). Through the iterative update steps of the Kalman filter,
the model parameters can be estimated by maximizing the likelihood
function. The log-likelihood function for parameter estimation is
calculated as follows:
$$l(y_{t_1}, y_{t_2},\ldots,y_{t_n}|\Theta, \mu, H, Z) = -\dfrac{(N+1)p}{2}\log (2\pi) -\dfrac{1}{2}\sum_{n = 0}^N \left(\log|F_{t_n}| + v_{t_n}^TF_{t_n}^{-1}v_{t_n}\right)$$
where the terms in the equation are defined as follows:
\begin{align*}
v_{t_n}&=y_{t_n} -\mu- Za_{t_{n}}\\
F_{t_n} &= ZP_{t_n}Z^T + H\\
P_{t_{n+1}} &=  C_{t_{n+1}}P_{t_n}(C_{t_{n+1}} - K_{t_{n}}Z)^T + Q_{t_{n+1}}\\
K_{t_n} &= C_{t_{n+1}} P_{t_n}Z^TF_{t_n}^{-1}\\
a_{t_{n+1}} &=  C_{t_{n+1}} a_{t_n} + K_{t_n}v_{t_n}\\
C_{t_{n+1}} &= e^{-\Theta\Delta t_{n+1}}
\end{align*}
Here, we assume $a_{t_0}= 0$ and $P_{t_0} = Q_{t_\infty}$. To perform
parameter estimation, the Kalman Filter was implemented in R, and the
optim function was used with the L-BFGS-B method, using numerical
approximations for the derivatives, to maximize the log-likelihood
function.

To ensure that the real parts of the eigenvalues are positive, we
impose the additional constraint that the diagonal entries of $\Theta$
are positive. This constraint is sufficient to ensure that the
eigenvalues have positive real part, but is not necessary. Thus, we
are restricting the space of possible estimates by imposing this
condition. However, most matrices whose eigenvalues have positive
real parts satisfy this condition, so the restriction is not too
severe.

To estimate the constrained parameters in the model, we reparametrize
some parameters to remove the constraints and simplify the
optimization process. In more detail, we estimate the log-transformed
parameters for the diagonal elements of $H$ and $\Theta$. To ensure
the diagonal elements of $\Theta$ are in decreasing order, we
reparametrize the diagonal elements of $\Theta$ as a sequence of
cumulative sums, i.e. $\theta_k=(\sum_k^m  \alpha_i)$, where
$\alpha_i, i=1,\ldots,m$ are all positive, and we estimate the
log-transformed $\alpha_i$'s.

These reparametrizations effectively convert the constrained
optimization problem into an unconstrained one, facilitating efficient
estimation while ensuring adherence to the model's structural
constraints. 

The code is available as an R package at
https://github.com/ShanglunLi/OUSSM.

\subsection{Model Selection}
\label{model}
Selecting an appropriate model dimension is crucial for balancing the
flexibility to model the complexity of the real underlying process with
the ability to estimate model parameters from the available data. In
this study, we determine the optimal state equation dimension using
the Akaike Information Criterion (AIC) and the Bayesian Information
Criterion (BIC), which are widely used for model comparison. These
criteria penalize model complexity to prevent overfitting while
maximizing the likelihood of the observed data. For our model, $AIC =
-2\ln (L) + 2k$ and $BIC = -2\ln (L) + \ln (n)k$, where $k = p(m + 2)
+ m(m+1)/2$ is the number of parameters in the model, $L$ is the
maximized value of the likelihood function, and $n$ is the number of
observations in the dataset. To determine the optimal model, we
compute AIC and BIC for state-space models of varying dimensions and
select the model with the lowest AIC and BIC scores.

\section{Simulation Study}
\label{sec4}
In this section, we present simulation studies on the factor OUSSM. We
examine the scenario where the OU process in the state equation is
two-dimensional, with the output equation being two, three, or
four-dimensional.

In our simulations, we consider 15 different $\Theta$ matrices, giving
a range of values of the matrix properties that might be expected to
affect the dynamics and the difficulty of the estimation; these
include the magnitude of the diagonal elements of $\Theta$; the
difference between the diagonal elements and the squared difference
between the eigenvalues, $(\lambda_1 - \lambda_2)^2$. For each output
dimension, we generate two $Z$ matrices: one with a small angle
between the columns of $Z$ and one with an orthogonal angle between
the columns of $Z$. This results in a total of 90 scenarios. Table~\ref{Zsim} and Table~\ref{Thetasim} in Appendix~C summarize the
different $\Theta$ and $Z$ values used in the simulations. In the
simulation, we set the sequence length, $N$, to be 5000 with time points equally spaced with time gap $\Delta t = 1$ and simulate
100 replicate datasets for each combination of $\Theta$ and $Z$
values.

When assessing performance, it is important to remember that we made an arbitrary choice of constraints to resolve the identifiability issues. Our measure of estimation performance should not depend on this arbitrary choice. From~\eqref{eq2}, for a given true matrix $\Theta$, any of the matrices $A\Theta A^T$ for an orthogonal matrix $A$, would give the same equation, still with $\Sigma = I$. Therefore, a natural choice of measure for performance accuracy is the minimum Frobenius norm difference between matrices of the form $A\Theta A^T$ and $B\hat{\Theta}B^T$ for orthogonal matrices $A$ and $B$. Fortunately, for the two-dimensional simulations, this coincides with the Frobenius norm of the difference between the chosen representatives.

\begin{theorem}
  For two-dimensional space, if $\Theta$ and $\hat{\Theta}$ are sums
  of antisymmetric and diagonal matrices with decreasing diagonal
  elements, and with $\Theta_{12}\hat{\Theta}_{12}>0$, then
  $\mathlarger{\inf}_{\stackrel{\Theta_1 \in S_1}{\Theta_2\in S_2}}
  \left\{\left\lVert\Theta_1 -\Theta_2\right\rVert_F\right\} =
  \left\lVert\Theta -\hat\Theta\right\rVert_F$, where $S_1 = \{A\Theta
  A^{T}|AA^T =I\} $ and $S_2 = \{B\hat\Theta B^{T}|BB^T =I\} $.
\end{theorem}

  From~\eqref{eq1}, $\exp(-\Theta)$ can be
  interpreted as autocorrelation of latent states at time step $\Delta
  t = 1$. Thus instead of calculating the absolute errors between the
  classes $\left\{A\Theta A^T\right\}$ and $\left\{B\hat\Theta
  B^T\right\}$ for any orthogonal matrices $A$ and $B$, we calculate
  the relative errors between two classes $\exp\left(-\left\{A\Theta
  A^T\right\}\right)$ and $\exp\left(-\left\{B\hat\Theta
  B^T\right\}\right)$.

  From the Cayley-Hamilton theorem, $X$ is a zero of its
  characteristic polynomial, thus for any two-dimensional matrix $X$,
  $X^2 = aX+bI$ for some constants $a$ and $b$. By recursion, it is
  easy to see $\exp(-X) = cX+dI$ for some constants $c$ and $d$. Thus,
  if $X$ is a sum of an antisymmetric matrix and a diagonal matrix, so
  is $\exp(-X)$. Furthermore, since $\exp\left(AX
  A^{-1}\right)=A\exp\left(X \right)A^{-1}$ for any matrix $X$, the
  minimum Frobenius norm between classes $\exp\left(-\left\{A\Theta
  A^T\right\}\right)$ and $\exp\left(-\left\{B\hat\Theta
  B^T\right\}\right)$ is the Frobenius norm of the error matrix, $E =
  \exp{(-\hat{\Theta})} - \exp{(-\Theta})$ (provided
  $\Theta_{12}\hat{\Theta}_{12}>0$, if not the minimum Frobenius norm
  is the Frobenius norm of $E' =
  \exp{(-\hat{\Theta}^T)} - \exp{(-\Theta})$). We define the Frobenius
  norm ratio of the error in the estimate
  $\exp\left(-\hat\Theta\right)$ relative to the truth as
$$\dfrac{\left\lVert
  E\right\rVert_F}{\left\lVert\exp{(-{\Theta})}\right\rVert_F} =\dfrac{\sqrt{\sum_{i,j}\left(E_{ij}\right)^2}}{\sqrt{\sum_{i,j}\left(\exp{(-{\Theta})}_{ij}\right)^2}}
$$

Figure~\ref{simu} presents the log-Frobenius norm ratio of the
estimation error versus the modulus of the larger eigenvalue,
$\operatorname{Mod}(\lambda_1)$, for different squared difference
between the eigenvalues, $(\lambda_1 - \lambda_2)^2$, dimension of
$y$ and whether the columns of $Z$ are orthogonal. Note that the
squared difference between the eigenvalues also indicates what sort of
dynamics occur in the system. If the eigenvalues are real, then
$\lambda_1-\lambda_ 2$ is real, so the square is positive, while if
$\lambda_1$ and $\lambda_2$ are complex, then they are conjugate, so
$\lambda_1-\lambda_ 2$ is imaginary, and the square is negative. From
Figure~\ref{simu}, we see that the most important factor influencing
the accuracy of parameter estimation is the modulus of the largest
eigenvalue of $\Theta$. Only in cases where this eigenvalue has
modulus less than 1 are we able to accurately estimate $\Theta$. This
makes sense, as large values of this eigenvalue mean a fast mean
reversion process. With fixed time lag between observations, the
autocorrelation between observations is relatively small, and the
measurement error is independent, so the problem becomes close to
unidentifiable, resulting in large errors in the
estimates.

Higher dimension of observations $y_t$ gives relatively better
estimation. This is because more dimensions in $y_t$ means more
information. However, this also increases the number of parameters in
$Z$ to be estimated, thus resulting in some additional challenge for
higher dimensional problems.

The angle between the columns of $Z$ does not appear to have a large
effect on the parameter estimation. The squared difference between
eigenvalues also does not appear to have a strong effect on the
estimation, both for eigenvalues being real or complex. Real and
complex eigenvalues produce qualitatively different dynamics, so it
would have been plausible to see an effect on estimation accuracy.

\begin{figure}[h]
\centering
\includegraphics[width = \textwidth]{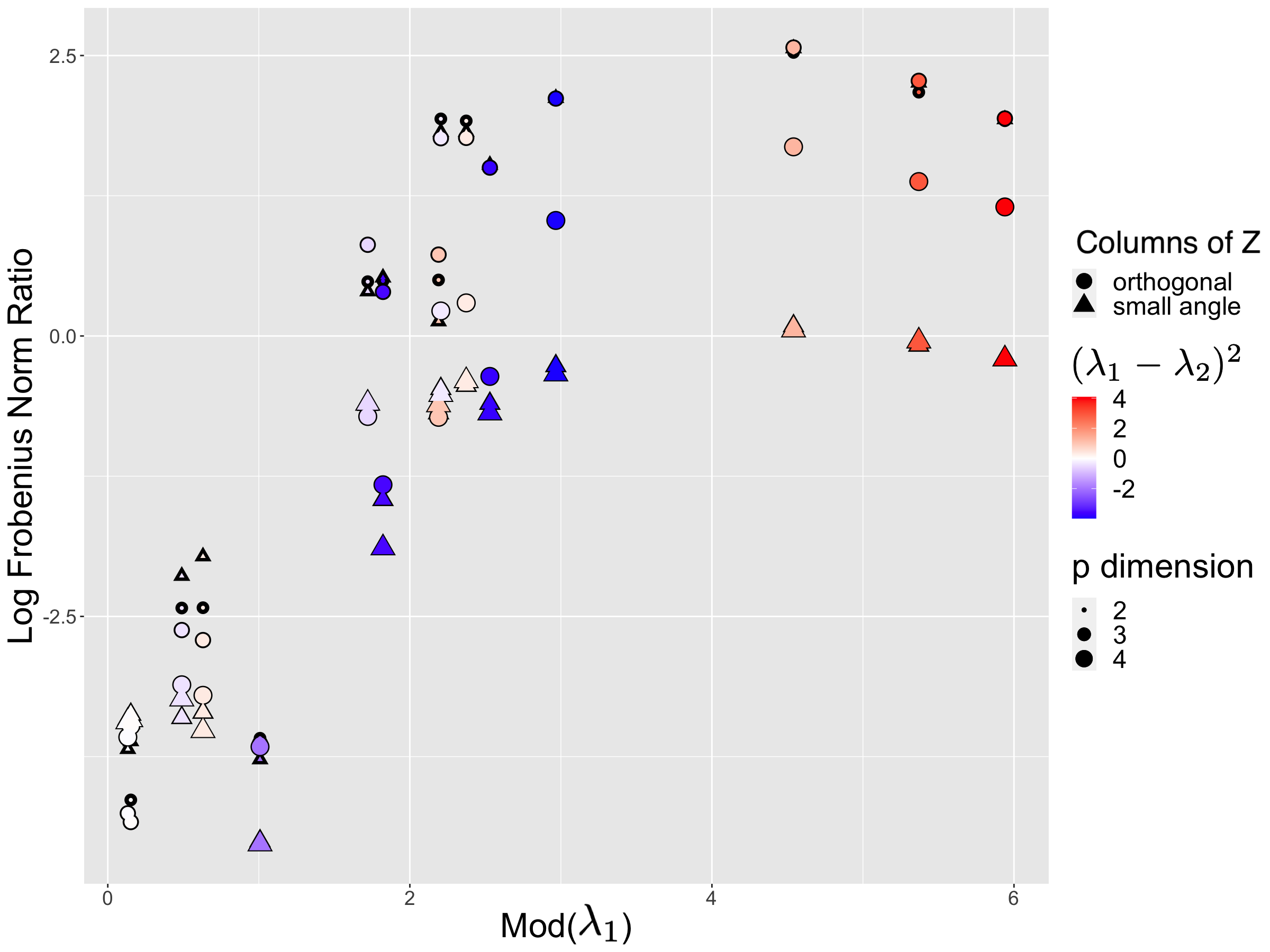}
\caption{Average Logarithm of Frobenius norm ratio between
  $\exp{(-\hat{\Theta})}$ and true $\exp{(-\Theta})$ value over 100
  simulations for each simulation scenario, compared with the scenario
  properties - modulus of the largest eigenvalue of $\Theta$ and
  squared difference between the eigenvalues of $\Theta$. Note that
  the magnitude of the squared difference between the eigenvalues
  indicates how close the eigenvalues are, while the sign of the
  squared difference indicates whether the eigenvalues are real or
  complex.}
\label{simu}
\end{figure}

For each simulation scenario, we compute the AIC for varying
dimensions of the state equation. Specifically, in each scenario, we
analyze 100 simulated datasets and determine the proportion of cases
where the minimum AIC corresponds to a one-dimensional,
two-dimensional, or higher-dimensional state equation. As shown in
Figure~\ref{AIC}, when the log Frobenius norm ratio is less than zero,
the probability of correctly identifying the true state equation
dimension is significantly higher. When the log Frobenius norm ratio
is larger than zero, a one-dimensional state equation is more often
chosen than higher dimensions. This shows that AIC balances the model
bias and model variance since the parameter variances in these cases
are larger. Using BIC (plot in Appendix~D), the model choice is more
often a one-dimensional state equation. In practice, we suggest
examining both information criteria to choose the dimension of the
state equation for OUSSM.

\begin{figure}[h]
\centering
\includegraphics[width = \textwidth]{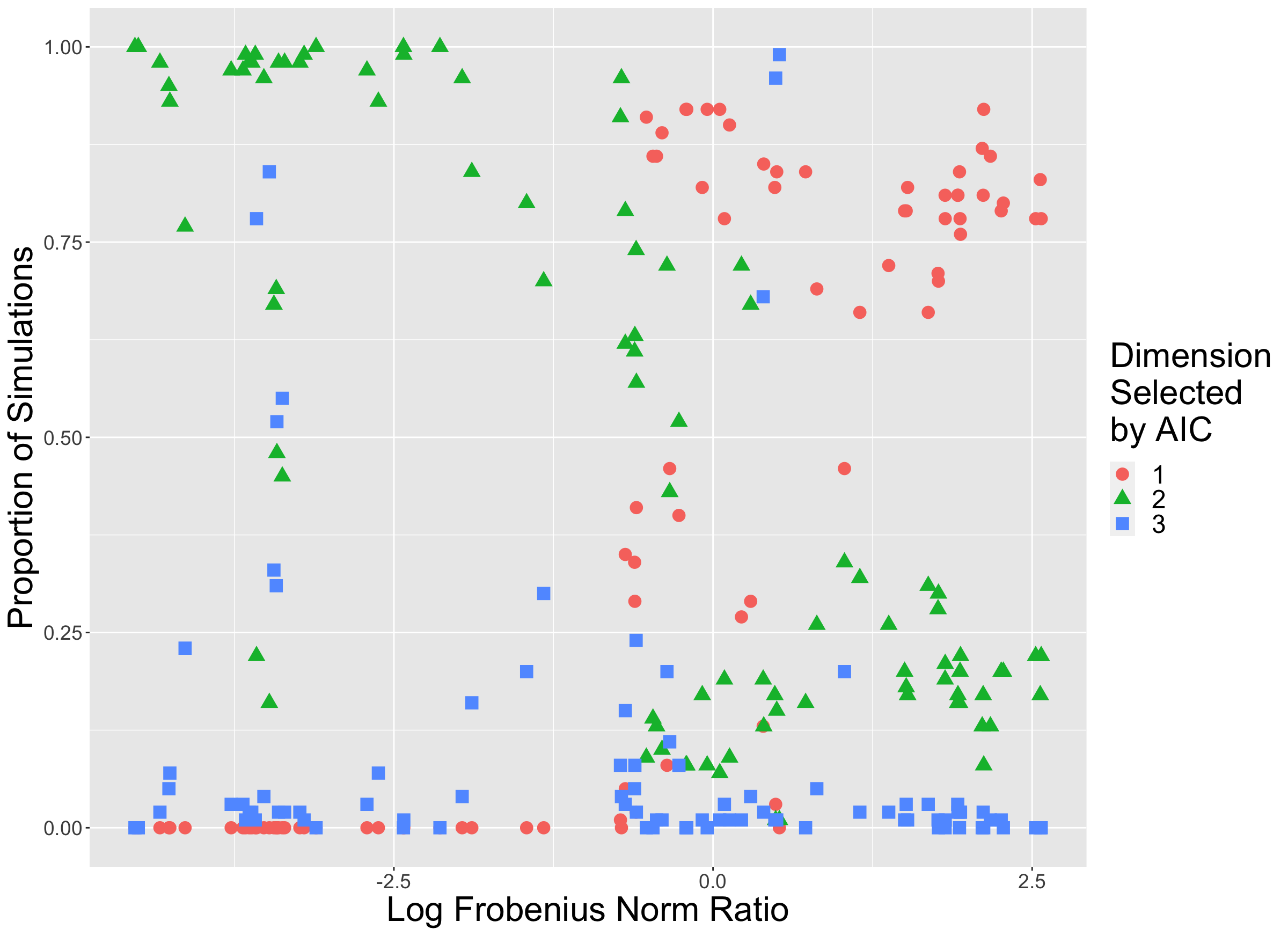}
\caption{Proportion of times AIC selects different dimensions for the state equation versus the average Logarithm of Frobenius norm ratio between $\exp{(-\hat{\Theta})}$ and true $\exp{(-\Theta})$ value over 100 simulations for each simulation scenario. }
\label{AIC}
\end{figure}

\section{Real Data Analysis}
\label{sec5}
\subsection{Microbiome Data}
We apply our method to the moving picture
dataset~\citep{caporaso2011moving} to learn more about the dynamics of the
human gut microbiome.
\subsubsection{Data Description and Preprocessing}
We use the factor OUSSM to study the gut microbiome data from a
healthy individual. (We focus on Person~2 in the moving picture study
because the time series is longer for this individual). The dataset
comprises 336 observations collected between October 2008 and January
2010. The sampling intervals were irregular, with some samples taken
daily and others taken at intervals of several days. The samples in
the Moving Picture dataset were sequenced using PCR targeting the V2
region of the 16S rRNA gene~\citep{costello2009bacterial}.

To properly analyze microbiome data, it is crucial to account for
sequencing depth. Microbiome data analysis presents challenges due to
the influence of technical factors on total read counts, which can
obscure associations with the true microbial abundances in the
original sample~\citep{zaheer2018impact}. Nevertheless, researchers
have found that relative abundances of microbial taxa provide a more
accurate representation of their proportions in the original
environment. As a result, various correction methods, based on
observed relative abundances, are commonly used to mitigate the
effects of these technical factors. In this study, we apply a
logarithmic transformation of the ratio of Operational Taxonomic Unit
(OTU) abundances between a specific genus and the total abundance of
phylum Bacteroidetes. We study the dynamics at genus level because the
low abundance at species level can cause the normal approximation to
be inaccurate.

One common challenge in microbiome data analysis is the presence of
excess zeros~\citep{kaul2017analysis}. To enable our log-scale analyses to work, we replace zero
counts with a pseudocount. As we are
focussing on relatively abundant genera, there are very few zero
counts, so the choice of pseudocount has little effect on the
estimated parameters. We use 0.3 for the pseudocount for this
analysis.

\subsubsection{Estimation results and interpretation for {\it Eubacterium} and {\it Clostridium}}

We examine the interactions between several different groups of genera
to demonstrate the types of dynamics that can be found within the
system.

For our first example, we analyze the genera {\it Eubacterium} and
{\it Clostridium} from the Tenericutes phylum. We fit the data
using state equations of varying dimensions and compute the
corresponding AIC and BIC values. As seen in Table~\ref{AICall}, both
AIC and BIC reach their minimum at a two-dimensional state equation,
indicating that a two-dimensional model provides the best balance
between model complexity and goodness-of-fit.

\begin{table}[h] 
\centering 
\caption{AIC and BIC values for different state equation dimensions
  for the sets of genera modelled in this paper. The lowest values are
  highlighted.}
\label{AICall} 
\begin{tabular}{ccccccc} 
\hline 
&State Dim ($m$) & 1 & 2 & 3 & 4 & 5\\\hline 
 {\it Eubacterium}  & AIC & 761.4 & \textbf{545.7} & 547.3 \\ 
and {\it Clostridium} & BIC & 787.4 & \textbf{586.5} & 606.8 \\ \hline 
 {\it Desulfovibrio} & AIC & 109.5 & 86.3 & \textbf{85.7} & 98.4 \\ 
and {\it Coprococcus}& BIC & 135.4 & \textbf{127.1} & 145.1 & 180.1\\ \hline
Proteobacteria &
AIC && 1818.5 & 1567.1 & \textbf{1516.4} & 1544.3  \\ 
phylum& BIC && 1889.1 & 1663.6 & \textbf{1642.6} & 1704.0 \\ \hline
Bacteroidetes &
AIC && 1701.8 & 1347.0 & \textbf{1173.7} & 1181.9  \\ 
phylum& BIC && 1787.2 & 1462.1 & \textbf{1322.2} & 1367.5 \\ \hline 
\end{tabular} 
\end{table} 

For this pair of genera, the eigenvalues of $\Theta$ are real, so as mentioned in Section~\ref{para_interpretation}, we can change the basis so that $\Theta$ is diagonal. For this changed basis, the parameter estimates are:
$$\Theta' = \begin{pmatrix}
1.0495 & 0 \\
0 & 0.0517 
\end{pmatrix}, Z' = \begin{pmatrix}
1.0060 & 0.1381 \\
0.3248 & 0.3095 
\end{pmatrix}, \Sigma' = \begin{pmatrix}
1 & -0.1620 \\
-0.1620 & 1 
\end{pmatrix}$$
These matrices indicate that there are two separate latent
mean-reverting processes with little interaction (low correlation in
the matrix $\Sigma$), with one exhibiting a faster rate of mean
reversion and the other a slower rate. {\it Eubacterium} is primarily
driven by the faster latent mean-reverting process, while {\it
  Clostridium} is influenced by both processes. To exemplify the
nature of the estimated dynamics, we simulate a sequence of relative
abundances using the estimated parameters, excluding the measurement
error. The simulated dynamics where {\it Eubacterium} exhibits a rapid
mean-reverting process, while {\it Clostridium} follows a combination
of this mean-reverting process with a relatively slower mean-reverting
process are illustrated in Figure~\ref{simuDynamic}.

This result aligns with previous research findings. For instance,
\citet{mukherjee2020gut} reported that short-chain fatty acid (SCFA)
levels significantly affect the abundance of {\it Eubacterium}, a
known SCFA-producing bacterium. Additionally, SCFA levels in the human
gut can fluctuate rapidly in response to dietary changes, particularly
within a few days of fiber and carbohydrate intake
adjustments~\citep{tan2023dietary,durholz2020dietary}. As SCFAs are
largely produced through bacterial fermentation of undigested
carbohydrates, increased fiber boosts SCFA
production~\citep{rios2016intestinal}. These findings could explain
the rapid mean-reverting behavior of {\it Eubacterium} in our results.

The rapid fluctuations in the abundance of {\it Clostridium} from the 
Tenericutes phylum are largely due to its responsiveness to
environmental and dietary changes. Research indicates that dietary
modifications, such as increased fiber intake or microbiota
transplants, can significantly affect the gut microbiota composition
within a matter of days, explaining the fast mean-reverting component
of {\it
  Clostridium}~\citep{zhong2021effect,karahan2024environmental}. This
could also explain why the faster mean-reverting component is shared
between {\it Clostridium} and {\it Eubacterium}. In addition, {\it
  Clostridium} species exhibit resilience and long-term
stability. After short-term disturbances, they tend to re-establish
due to microbial diversity, functional redundancy, and interactions
that promote ecosystem stability. Factors such as microbial
competition and biofilm formation might also contribute to the slow
mean-reverting component of {\it
  Clostridium}~\citep{fassarella2021gut}.

\begin{figure}[h]
  \centering
  \begin{subfigure}{\textwidth}
  \centering
    \includegraphics[width=0.65\textwidth]{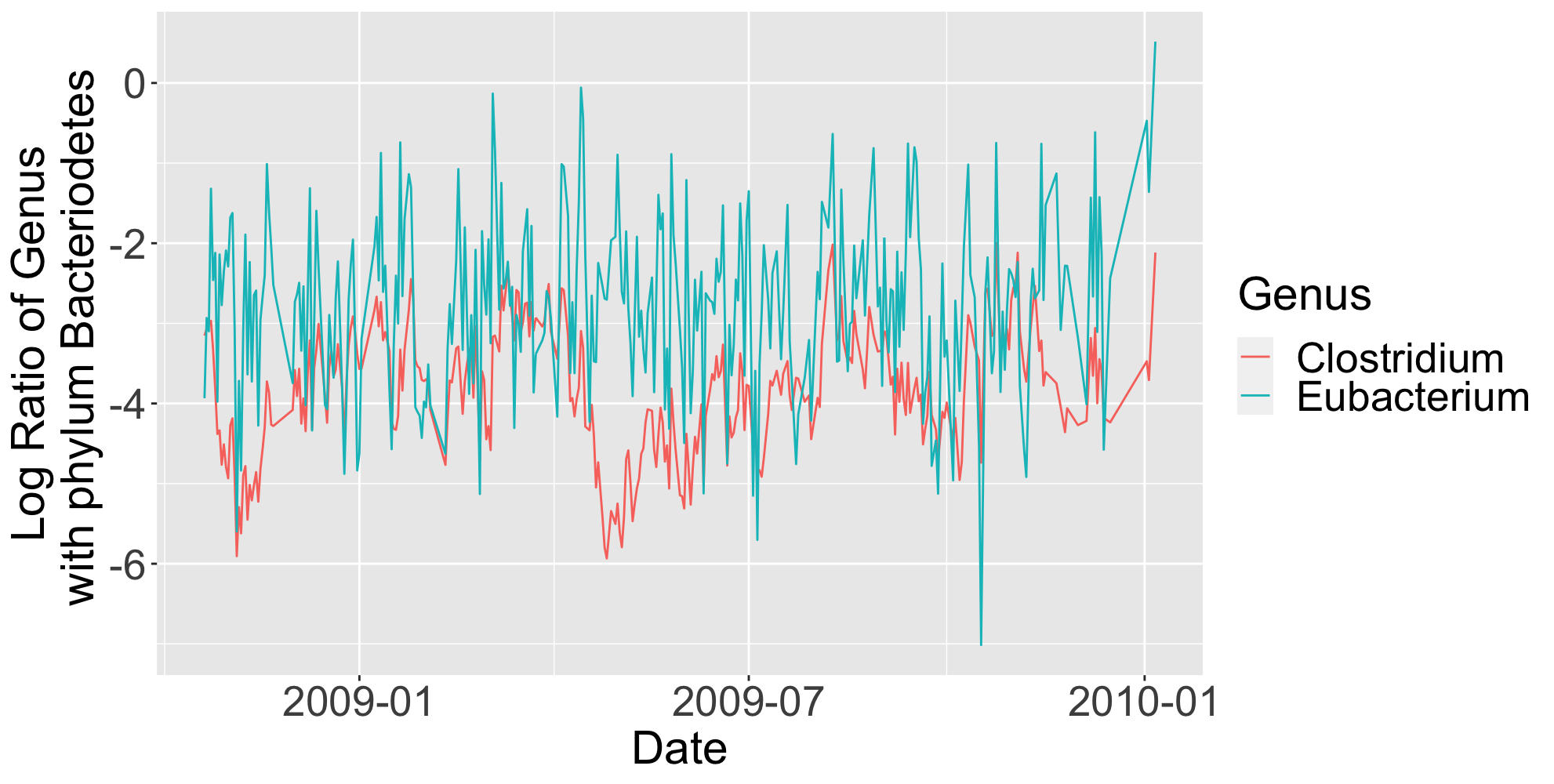}
    \caption{Simulated dynamics based on estimated parameters}
\label{simuDynamic}
\end{subfigure}
  \begin{subfigure}{\textwidth}
\includegraphics[width=0.5\textwidth]{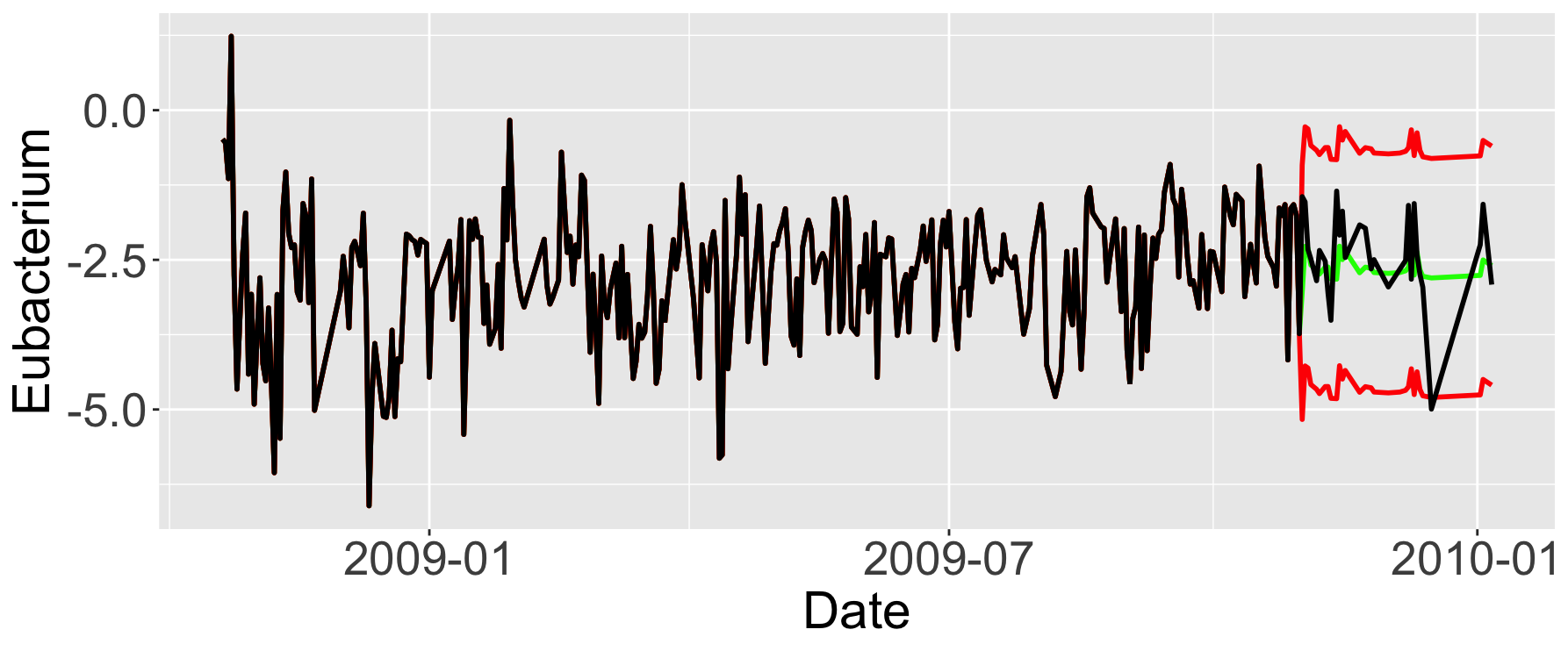}
\includegraphics[width=0.5\textwidth]{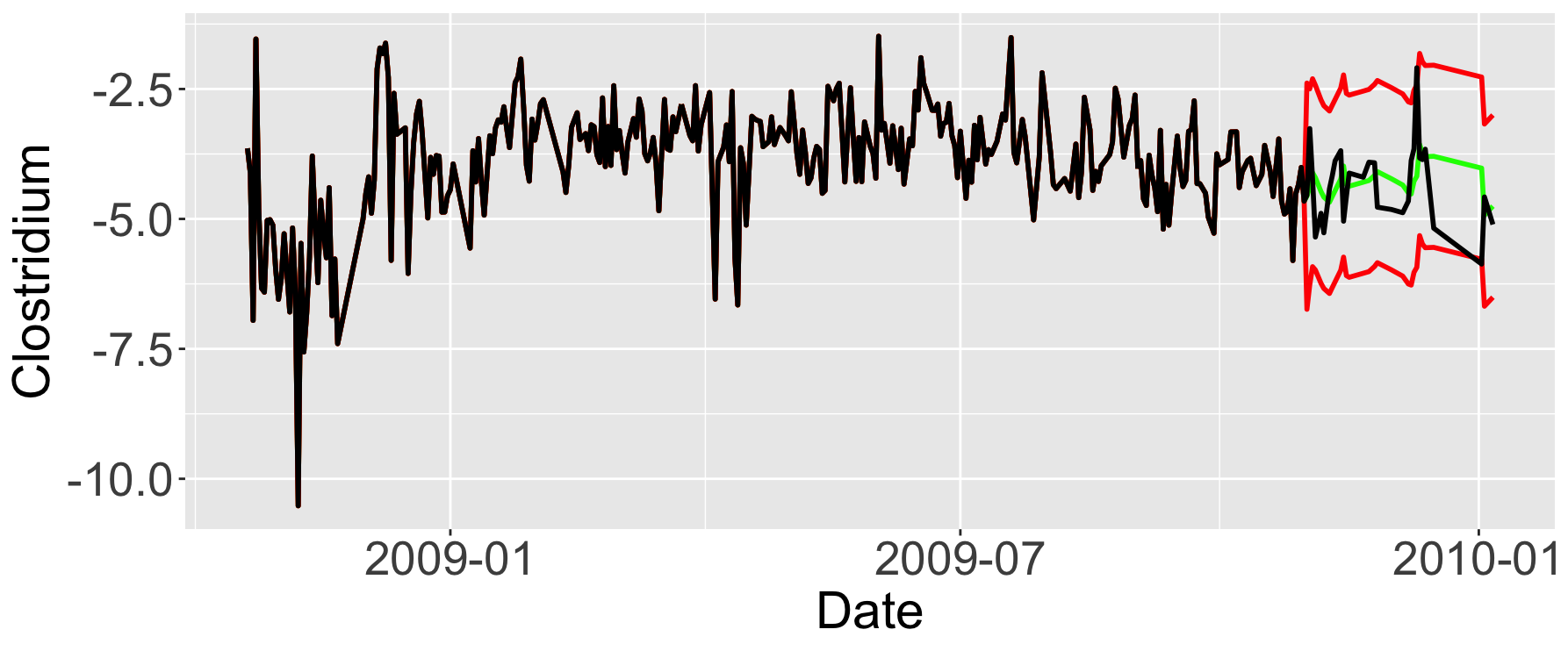}
\caption{Prediction based on factor OUSSM. Black line - observed data. Green line - one-step forward predictions using parameter values estimated on the training data. Red line - 95\% prediction intervals.}
\label{predReal}
\end{subfigure}
    
\caption{Simulated latent dynamics and forecasts for {\it Eubacterium} and {\it Clostridium}}
\end{figure}

To assess the reliability of the estimated parameters, we use the
first 90\% of the data as a training set and the remaining 10\% as a
test set. We use the training data to estimate the parameters, then
use the estimated parameters to make one-step forward predictions for
the remaining timepoints. Specifically, we apply the Kalman filter to
our test data, forcasting each $x_{t_n}$ at the test timepoints and
using the resulting values with our parameter estimates to predict the
corresponding $y_{t_n}$. Additionally, we calculate a 95\% prediction
interval for each prediction. Figure~\ref{predReal} compares the
observed abundances with these predictions and a 95\% prediction
interval. We see that the majority of points fall within the 95\%
prediction intervals.

\subsubsection{Estimation results and interpretation for {\it Desulfovibrio} and {\it Coprococcus}}
For an example of qualitatively different dynamics, we look at the
genera {\it Desulfovibrio} and {\it Coprococcus}. As before, we fit
the data using state equations of varying dimensions and compute the
corresponding AIC and BIC values. %
%
From Table~\ref{AICall}, we observe that AIC reaches its minimum at
$m=3$, while BIC is minimized at $m=2$. However, the AIC value for
$m=2$ is very close to that of $m=3$, whereas the BIC at $m=3$ is
significantly larger than at $m=2$. Given that the AIC difference
between $m=2$ and $m=3$ is relatively small, we determine that a
two-dimensional state equation provides the most balanced model
selection. For these genera, the estimated $\Theta$ has complex
eigenvalues.
$$\Theta = \begin{pmatrix}
0.9883 & 0.1981 \\
-0.1981 & 0.6960 
\end{pmatrix}, Z = \begin{pmatrix}
0.3588 & 0.6064 \\
-0.1747 & 0.6417
\end{pmatrix},\Sigma = \begin{pmatrix}
1 & 0 \\
0 & 1 
\end{pmatrix}$$

\begin{figure}[H]
\begin{subfigure}{\textwidth}
  \centering
\includegraphics[width=0.6\textwidth]{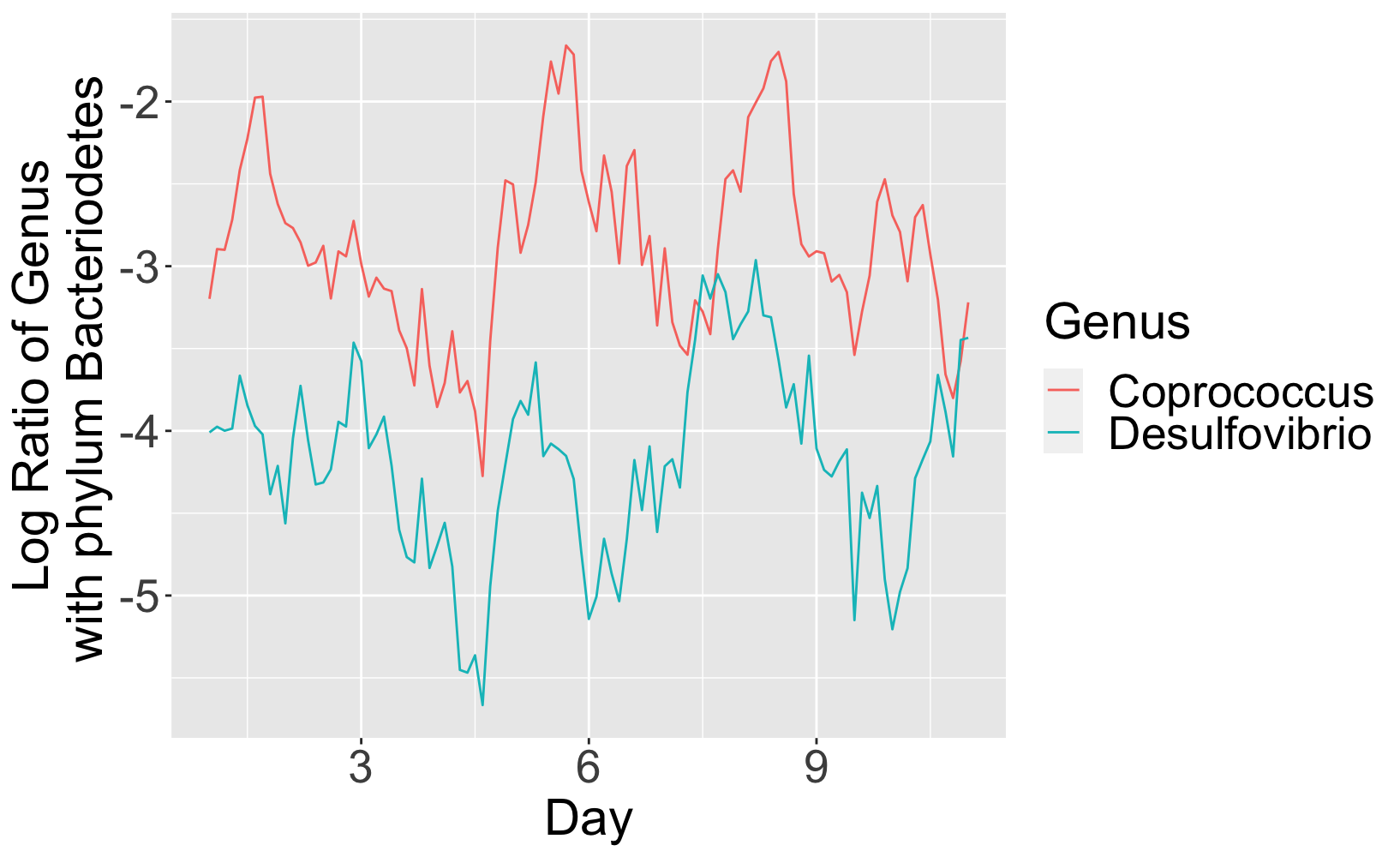}\\
\includegraphics[width=0.4\textwidth]{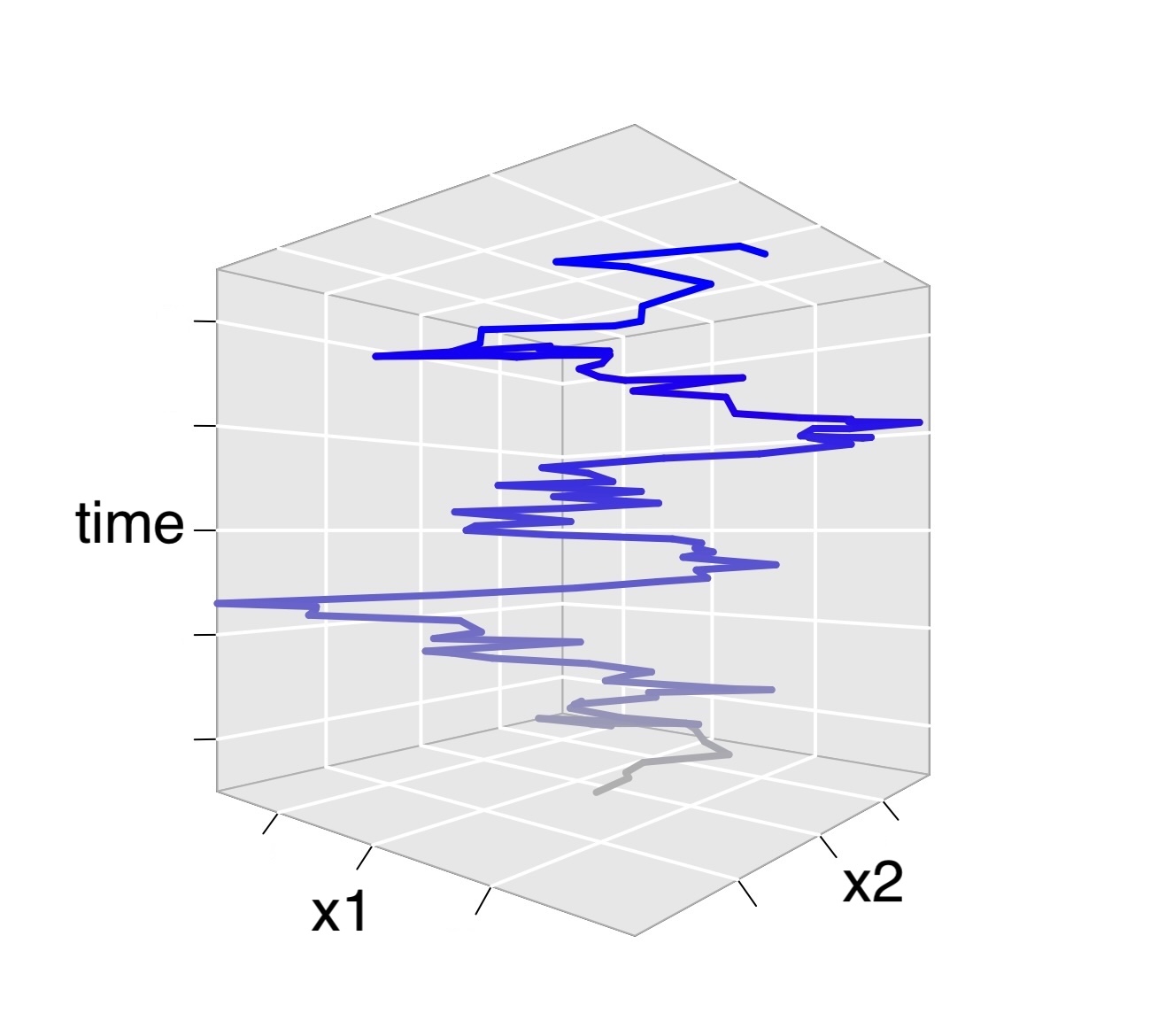}
\includegraphics[width=0.55\textwidth]{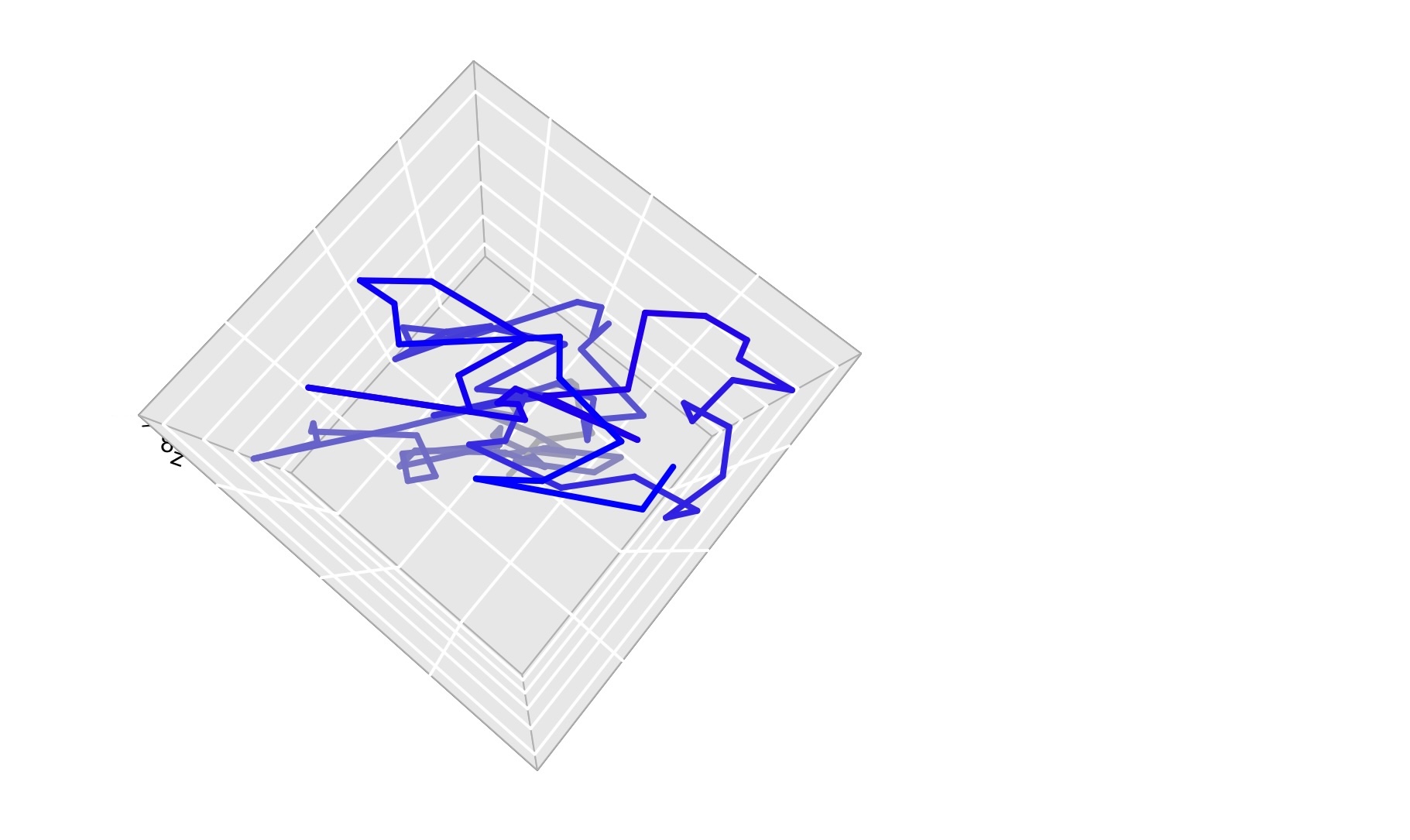}
\caption{Simulated latent dynamics under the estimated  parameters}
\label{simuComplex}
\end{subfigure}

\begin{subfigure}{\textwidth}
\centering
\includegraphics[width=0.45\textwidth]{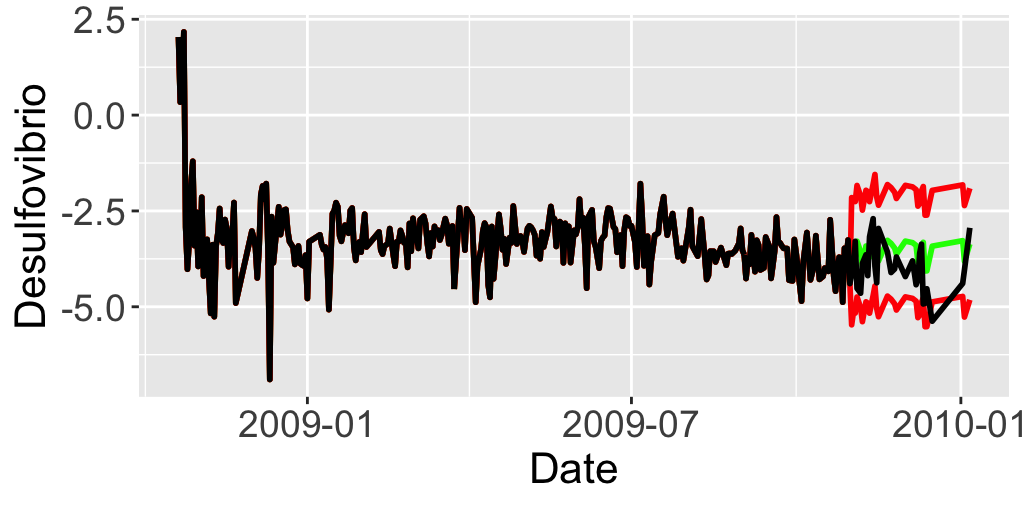}
\includegraphics[width=0.45\textwidth]{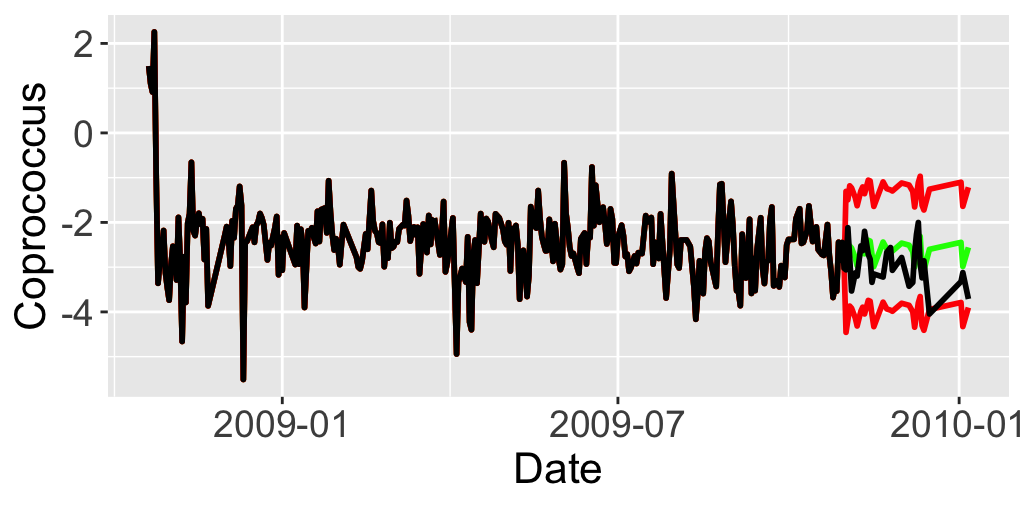}
\caption{Prediction based on factor OUSSM. Black line - observed data. Green line - one-step forward predictions using parameter values estimated on the training data. Red line - 95\% prediction intervals.}
\label{predComplex}
\end{subfigure}
\caption{Results for {\it Desulfovibrio} and {\it Coprococcus}}

\end{figure}

As for {\it Eubacterium} and {\it Clostridium}, we simulate abundances
(without measurement error) based on the estimated dynamics, to gain
better insight into the estimated dynamics. Figure~\ref{simuComplex}
shows a simulated set of dynamics from these parameters. We see that
our estimated parameters show a much larger correlation between
abundances of these genera. As always for complex eigenvalues, there
is a slight lag in the relation, with {\it Coprococcus} peaking
slightly after {\it Desulfovibrio}. This type of lagged correlation is
typically the sort of dynamics that happen in a predator-prey
system. It can also arise from certain combinations of competition and mutualism. 

Furthermore, an examination of the correlation
between the two simulated microbial taxa, {\it Desulfovibrio} and {\it
  Coprococcus}, reveals a positive correlation. This finding aligns
with the observations reported by \citet{chen2021desulfovibrio},
supporting the biological relevance of the inferred relationships. The
ability of our model to capture such correlations suggests that it can
provide meaningful insights into microbial interactions, further
underscoring its utility for microbiome research and related inference
tasks.

Figure \ref{predComplex} shows the one-step forward predictions with
prediction intervals based on the estimated parameters using the
initial 90\% of the data as training data. Again, we see that the data
mostly fall within the prediction intervals.

\subsubsection{Estimation results and interpretation for Different Phyla}
\paragraph{Proteobacteria Phylum}

Applying our estimation method to selected genera within the 
Proteobacteria phylum, specifically \textit{Desulfovibrio},
\textit{Bilophila}, \textit{Campylobacter}, and \textit{Escherichia},
we determine that the state equation achieves the minimum AIC at four
dimensions as shown in Table~\ref{AICall}.

After a transformation such that $\Theta$ is block diagonal, the
estimated parameters are

$$\Theta = \begin{pmatrix}
  1.2406& -0.1660& 0.0000 &0.0000\\
  0.1660 & 1.2406& 0.0000& 0.0000\\
  0.0000 & 0.0000 &0.5500 &0.0000\\
  0.0000 & 0.0000& 0.0000 &0.0738\\
\end{pmatrix}$$ 
$$Z = \begin{pmatrix}
  1.8010 &0.3737 & 0.2199 & 0.7495\\
  1.9865 &0.4486 & 0.1538 & 0.3428\\
  -4.9384 &4.8996 & 2.3251& -0.2312\\
  -1.1243 &2.3004 &-0.9356 & 0.3946\\
\end{pmatrix}$$
$$ \Sigma = \begin{pmatrix}
  0.3343  &0.0811 & 0.1412& -0.2601\\
  0.0811 & 0.7899& -0.5903& -0.0204\\
  0.1412 &-0.5903 & 0.8958 &-0.1575\\
  -0.2601 &-0.0204 &-0.1575&  0.9800\\
\end{pmatrix}$$

From the loading matrix, both \textit{Desulfovibrio} and
\textit{Bilophila} are dominated by the first latent process, with
each of the four genera being a mixing of all four latent
processes. The four latent OU processes are moderately correlated,
with only one relatively larger negative correlation between the
second and third latent processes. The presence of complex eigenvalues
indicates a spiral structure in the system's dynamics, suggesting fast
oscillatory behavior in the temporal evolution of microbial abundances
of all these four genera. The two real eigenvalues correspond to one
moderate and one very slow mean-reverting component. This structure
aligns with the expectation that microbial populations are influenced
by both cyclic interactions and stabilizing regulatory forces.

\paragraph{Bacteroidetes Phylum}

A similar estimation procedure is applied to genera within the 
Bacteroidetes phylum, including \textit{Porphyromonas},
\textit{Prevotella}, \textit{Parabacteroides}, \textit{Alistipes}, and
\textit{Odoribacter}. The optimal state equation dimension is again
found to be four, as indicated by the minimum AIC criterion shown in
Table~\ref{AICall}.

The estimated parameters, after a transformation such that $\Theta$ is
block diagonal, are

$$\Theta = \begin{pmatrix}
  2.1197& 0.2790& 0.0000 & 0.0000\\
  -0.2790 & 2.1197 &0.0000&  0.0000\\
  0.0000 & 0.0000 &0.0541& 0.0580\\
  0.0000 & 0.0000 &-0.0580 & 0.0541\\
\end{pmatrix}$$ 
$$Z = \begin{pmatrix}
  6.5600 &0.0777& 1.0319& 0.1752\\
  7.3411& 0.8014& 1.1135& 0.2286\\
  -1.9118& 1.7643 & -0.4745& 0.1113\\
  -2.0149 &1.6952 & -0.2674 &0.2815\\
  -1.9891& 1.7579 & -0.6163& 0.7251\\
\end{pmatrix}$$
$$ \Sigma = \begin{pmatrix}
  0.5103 &0.0931 &-0.0687 & -0.0011\\
  0.0931 & 0.4901&  -0.0853 &-0.0713\\
  -0.0687 & -0.0853 & 0.4681& 0.2530\\
  -0.0011 &-0.0713 &0.2530 & 0.5315\\
\end{pmatrix}$$

The first two latent processes are almost independent of the last two
latent processes with the first two latent processes exhibiting fast
mean reversion, with some cyclical components and the last two latent
processes slowly oscillating. From the loading matrix, these five
genera all have the largest loading on the first latent process,
showing they all are dominated by a fast oscillatory behavior in the
temporal dynamics and all also have a slow oscillatory component. The
magnitudes of the coefficients in the loading matrix show that the
dynamics of these five genera are all correlated.

\subsection{Sea Surface Temperature Data}
\label{SST}
In this section, we apply our method to the sea surface temperature dataset. 
\subsubsection{Data Description and Preprocessing}
The sea surface temperature (SST) data used in this study were obtained from the ERA5 reanalysis produced by the Copernicus Climate Change Service (C3S) at the European Centre for Medium-Range Weather Forecasts (ECMWF) \citep{C3S2017}. ERA5 provides a consistent global record of atmospheric, land, and oceanic variables from 1940 to the present, combining numerical weather prediction model outputs with a wide range of historical observations through a data assimilation framework. The SST variable represents the temperature of the ocean surface layer, expressed in kelvins, and is derived primarily from satellite and \textit{in situ} measurements assimilated into the reanalysis system.

For this analysis, we extracted the ERA5 daily data at 12:00~UTC from the single-level field ``sea surface temperature.'' The data have a spatial resolution of $0.25^\circ \times 0.25^\circ$ (approximately 25~km at the equator) and cover all ocean grid points between $90^\circ$S and $90^\circ$N. To focus on the study region of interest, we subset the global field to the North Atlantic sector (approximately $33^\circ$--$37^\circ$N, $55^\circ$--$59^\circ$W) with a spatial sampling resolution of $1^\circ$ as shown in Figure~\ref{SST_map}. This region encompasses the area relevant to our subsequent analyses.

\begin{figure}[h]
\centering
\includegraphics[width=\textwidth]{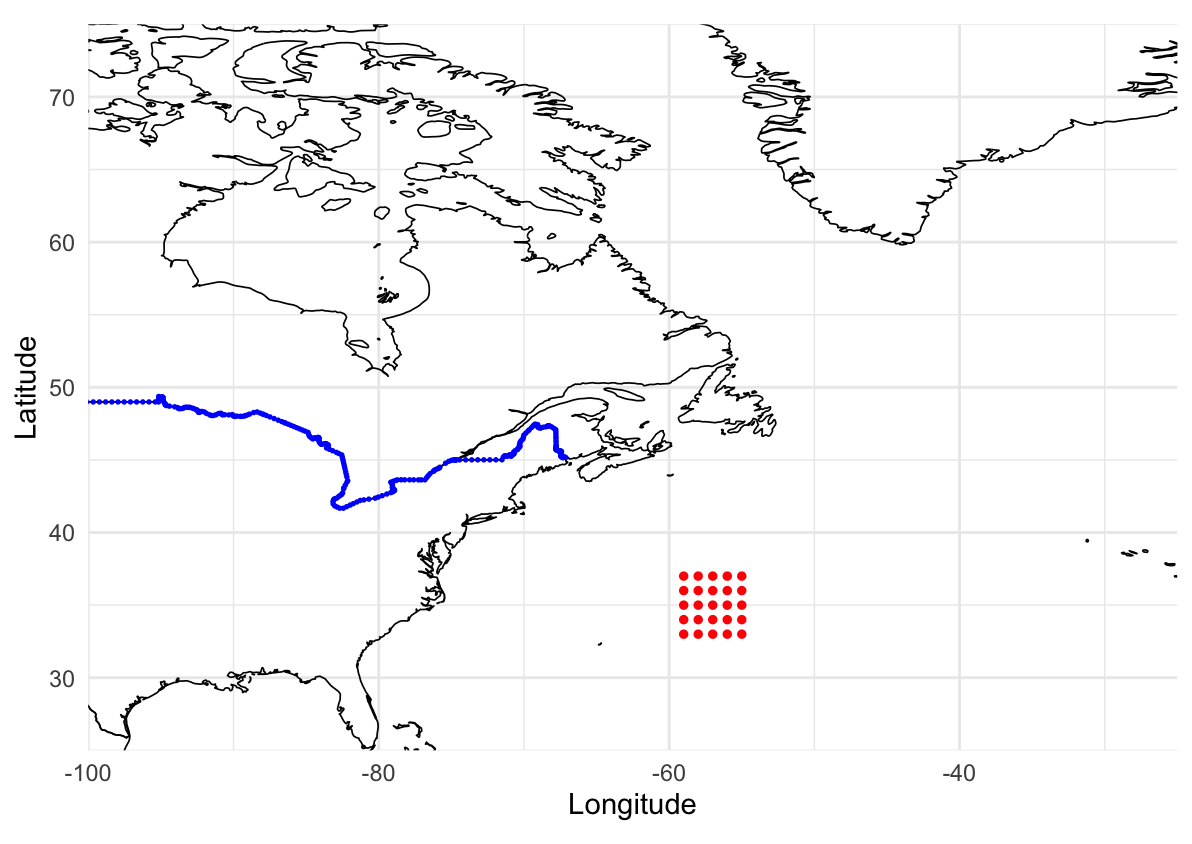}
\caption{Geographical region of sea surface temperature (SST) grid points used in the analysis. Each red dot represents a selected ERA5 grid point within the North Atlantic region of interest. The black outlines indicate the coastlines of North America and adjacent areas. The blue line highlights the international boundary between Canada and the United States.}
\label{SST_map}
\end{figure}

Each SST value represents the model-estimated instantaneous temperature at the ocean surface at a given grid point and time. The dataset is distributed through the Copernicus Climate Data Store (CDS) and is freely available under the Copernicus open data policy.

Before statistical modeling, the SST data were preprocessed to remove large-scale seasonal variation and ensure approximate stationarity of the time series. For each selected grid point, we computed a multi-year day-of-year (DOY) climatology to characterize the mean seasonal cycle of SST across the five-year study period (2020--2024). This climatological component was then subtracted from the raw daily SST to obtain the deseasonalized anomaly series, defined as
\[
Y_t = \mathrm{SST}(t) - \mathrm{climatology}_{\mathrm{DOY}}(t).
\]
The resulting anomalies fluctuate around zero and primarily capture short-term stochastic fluctuations in SST rather than the deterministic annual cycle. These anomaly time series were subsequently used as input to the OUSSM framework.

As shown in Table~\ref{AIC_SST}, both the AIC and BIC attain their minimum values at state dimension $m = 11$.

\begin{table}[h] 
\centering 
\caption{AIC and BIC values for different state equation dimensions. The lowest values are highlighted.} 
\label{AIC_SST} 
\begin{tabular}{cccc} 
\hline 
State Dimension ($m$) & 10 & 11 & 12 \\\hline 
AIC & -107494.1 & \textbf{-108689.5} &  -107219.3  \\ 
BIC & -105537.9 & \textbf{-106534.9} & -104860.8 \\ \hline 
\end{tabular} 
\end{table} 

The estimated parameters (after transforming so that $\Theta$ is block diagonal) are
$$
\Theta = \text{diag}\left(
  0.3684, \; 
  0.289, \; 
  0.2335, \; 
  \mathbf{B}_1, \; 
  \mathbf{B}_2, \; 
  \mathbf{B}_3, \; 
  0.0867, \; 
  0.0365
\right)
$$
where the $2 \times 2$ rotation blocks are defined as:
$$
\mathbf{B}_1 = \begin{pmatrix} 0.1873 & -0.0411 \\ 0.0411 & 0.1873 \end{pmatrix}, \quad
\mathbf{B}_2 = \begin{pmatrix} 0.1581 & -0.0143 \\ 0.0143 & 0.1581 \end{pmatrix}, \quad
\mathbf{B}_3 = \begin{pmatrix} 0.0848 & -0.0311 \\ 0.0311 & 0.0848 \end{pmatrix}
$$

Each row of the loading matrix $Z$ corresponds to an SST location, ordered systematically by latitude and longitude: rows~1--5 represent longitudes $59^\circ$W to $55^\circ$W at $37^\circ$N, rows~6--10 the same longitudes at $36^\circ$N, and so on down to rows~21--25 at $33^\circ$N. Within each group of five rows, the first row represents the westernmost point.

The block--diagonal structure of $\Theta$ shows that the latent system consists of a mixture of real and complex eigenmodes. The real eigenvalues range from 0.368 to 0.0365, corresponding to a half-life range from approximately 2 to 19 days. The imaginary parts of the eigenvalues are between 0.014 and 0.044, corresponding to periodic behaviour with a period between 143 and 440 days. This spectrum indicates that SST variability in the region is driven by both fast and slow dynamical components, including oscillatory modes that may represent slowly evolving basin-scale adjustments.

\begin{figure}[H]
    \centering

    \begin{subfigure}[t]{0.85\textwidth}
        \centering
        \includegraphics[width=\textwidth]{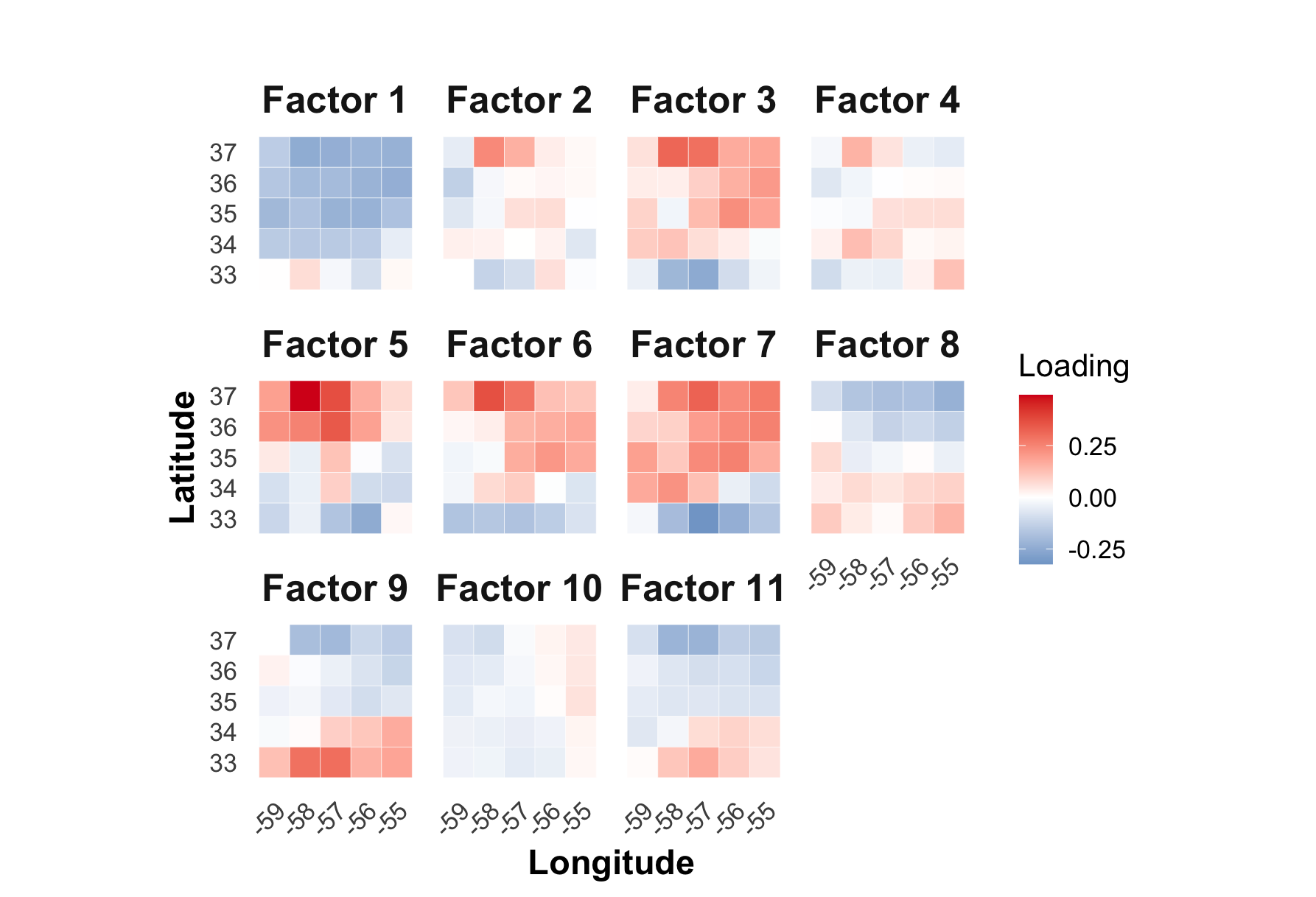}
        \caption{Spatially mapped loadings of the transformed $Z$ matrix. }
        \label{fig:Z_heatmap}
    \end{subfigure}
    \hfill\break
    \begin{subfigure}[t]{0.6\textwidth}
        \centering
        \includegraphics[width=\textwidth]{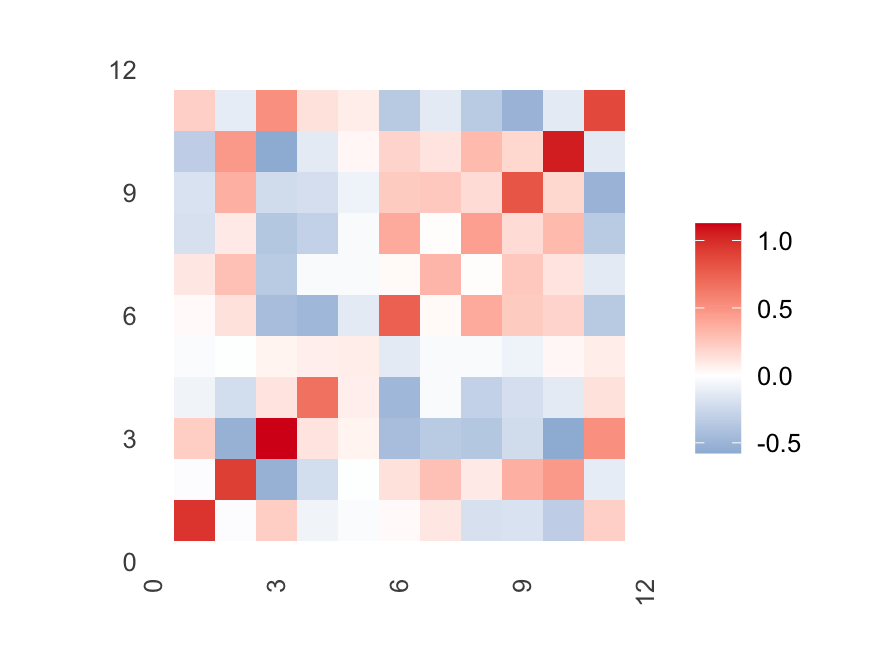}
        \caption{Heatmap of the transformed $\Sigma$ matrix.}
        \label{fig:Sigma_heatmap}
    \end{subfigure}

    \caption{
        Visualization of the estimated model matrices. 
        (a) The loading matrix $Z$ is reshaped to display the spatial pattern for each of the 11 latent factors, mapped to their corresponding Longitude ($55^{\circ}$W--$59^{\circ}$W) and Latitude ($33^{\circ}$N--$37^{\circ}$N). 
        (b) The diffusion-error covariance matrix $\Sigma$ illustrates the dependencies between the latent factors. 
        In both figures, red indicates positive entries and blue indicates negative entries.
    }
    \label{fig:Z_Sigma_heatmaps}
\end{figure}

The spatially mapped loading matrix $Z$ (Figure \ref{fig:Z_heatmap}) reveals that the latent factors correspond to distinct spatial modes defined by a latitudinal contrast. Specifically, Factors 3, 5, 6, 7, 8, 9 and 11 exhibit a North-South dipole structure, characterized by positive loadings in the northern latitudes ($35^\circ$N--$37^\circ$N) and negative loadings in the south ($33^\circ$N) or the reverse. The magnitude of these loadings is generally higher for longitudes in the centre of the region, with points on the east and west boundary of the region being less influenced by these factors. Factor 3 and 11 correspond to real eigenvalues, while Factors 6 and 7, and Factors 8 and 9 are $2\times 2$ blocks. This spatial separation suggests that the OUSSM has successfully disentangled distinct physical processes---likely separating variability driven by southern boundary currents from those driving northern sub-basin scale adjustments---consistent with the view that North Atlantic SST reflects a superposition of multiple multi-scale processes \citep{grossmann2009review,xu2023multi}.

The diffusion-error covariance matrix $\Sigma$ (Figure \ref{fig:Sigma_heatmap}) shows moderate covariances among the latent factors, including both positive and negative off-diagonal entries. Because $\Sigma$ characterizes the covariance structure of the stochastic forcing in the continuous-time OU dynamics, these values indicate coupling in the excitation of the latent processes themselves. Positive covariances suggest that some modes are driven coherently by shared environmental influences, whereas negative covariances may reflect compensatory or antagonistic forcing relationships. The lack of a strong block-diagonal pattern implies that the stochastic forcing of the components is only partially separable and incorporates mixed large-scale influences.

Such behaviour is consistent with studies showing that North Atlantic SST variability is shaped by partially coherent basin-scale processes rather than strictly independent modes (e.g., \citet{buckley2014low}). Taken together, the estimated $\Theta$, $Z$, and $\Sigma$ matrices indicate that SST variations in this region arise from a mixture of interacting dynamical components whose spatial signatures overlap, and that the stochastic forcing driving these modes is moderately shared across factors.

\section{Conclusion}
\label{sec6}
In this paper, we introduced a factor Ornstein-Uhlenbeck State Space
Model (OUSSM) to address the limitations of the traditional OU
process, particularly in handling measurement error and interactions
across multiple dimensions. Our model extends the conventional OUSSM
framework, which has been predominantly used for one-dimensional
problems, to accommodate more complex, real-world scenarios where
multiple interacting variables must be considered simultaneously. One
of the key contributions of this work is the reparametrization to
avoid identifiability problems in the model.

Our simulation studies demonstrated the efficacy of the proposed
model, particularly in its ability to accurately estimate key
parameters under a variety of scenarios. We see that when the
eigenvalues of the mean-reversion matrix are small, relative to the
time gap between samples, we are able to estimate the model parameters
fairly accurately. When one of the eigenvalues is larger, it becomes
more difficult to estimate the parameters.

We applied our model to the gut microbiome of a healthy individual. We
were able to identify qualitatively different types of latent temporal
dynamics together with the interactions between different genera. This
yields significant improvements in our understanding of microbial
dynamics.

We also applied the model to a sea-surface temperature (SST) dataset
from the western North Atlantic, demonstrating its usefulness beyond
biological systems. In this setting, the model uncovered multiple latent
mean-reverting components operating at distinct temporal scales, ranging
from fast synoptic variability to slower basin-scale thermal patterns.
The spatial loading structure revealed systematic longitudinal gradients
in how these latent processes influence SST, while the covariance
structure captured coherent large-scale anomalies as well as
compensatory warm--cool interactions. Together, these results highlight
the model's flexibility in representing complex spatio-temporal
processes in environmental systems.

An important topic for further research is how the sampling design
impacts the ability of the model to estimate model parameters,
particularly the matrix $\Theta$. As collection of samples is often
challenging, it is important to make sure each sample gives the most
information possible about the dynamics.

\section*{Disclosure statement}\label{disclosure-statement}

The authors have no conflicts of interest to declare.

\section*{Data Availability Statement}\label{data-availability-statement}

The human gut microbiome data used in this study are available from the ``Moving Pictures of the Human Microbiome'' project \citep{caporaso2011moving}. The Sea Surface Temperature (SST) data are available from the Copernicus Climate Change Service (C3S) Climate Data Store (CDS)\citep{C3S2017}.

\newpage
\bibliography{bibliography.bib}

\newpage
\appendix
\section{Proof of Theorem 1}
\label{AppA}

\setcounter{theorem}{0}

\begin{theorem}
  For any real matrices $\Theta$ and $\Sigma$, where $\Sigma$ is
  symmetric and positive definite, there is an invertible real matrix
  $A$, such that $A\Sigma A^T = I$, and $\tilde \Theta = A \Theta
  A^{-1}$ is a sum of an antisymmetric matrix plus a diagonal matrix
  with decreasing diagonal elements. Furthermore, if the diagonal
  elements of $\tilde \Theta$ are distinct, then $\tilde \Theta$ is
  unique up to conjugation by a diagonal matrix whose square is the
  identity matrix. In this case, the matrix $A$ is also unique up to
  left multiplication by a diagonal matrix with square the identity.
\end{theorem}

\begin{proof}
First of all, since $\Sigma$ is symmetric and positive definite, there
is a unique positive definite choice for $\Sigma^{-1/2}$. Therefore,
for any orthogonal matrix $B$, $A = B\Sigma^{-1/2}$ satisfies $A
\Sigma A^T = B\Sigma^{-1/2} \Sigma (\Sigma^{-1/2})^T B^T = I$.
  
Let $\Phi = \Sigma^{-1/2} \Theta\Sigma^{1/2}$. Then, $\Phi + \Phi^T$
is symmetric, so it is diagonalizable over $\mathbb{R}$ by orthogonal
matrices. Let $B$ be an orthogonal matrix, such that $D = B\left(\Phi
+ \Phi^T\right)B^{-1}$ is diagonal with decreasing entries. Then, $D =
B\Phi B^{-1} + B\Phi^TB^{-1} = B\Phi B^{-1} + (B\Phi B^{-1})^T$. Now,
if we let $A = B\Sigma^{-1/2}$, then we have shown that $A\Sigma A^T =
I$ and $\tilde \Theta = A\Theta A^{-1} = B\Phi B^{-1} =
\frac{1}{2}(\tilde \Theta + \tilde \Theta^T + \tilde \Theta - \tilde
\Theta^T) =\frac{1}{2}(D + \tilde \Theta - \tilde \Theta^T) $, which
is an antisymmetric matrix plus a diagonal matrix with decreasing
diagonal entries.

To prove uniqueness, let $A_1$ and $A_2$ be invertible real matrices
that satisfy the conditions. Let $\tilde \Theta_1 = A_1 \Theta
A_1^{-1}$ and $\tilde \Theta_2 = A_2 \Theta A_2^{-1}$. Suppose the
diagonal elements of $\tilde \Theta_1$ are distinct, we need to show
$\tilde \Theta_1 = \tilde \Theta_2$.

We know $A_1\Sigma A_1^T = I = A_2\Sigma A_2^T$, so $A_1\Sigma^{1/2}$ and $A_2\Sigma^{1/2}$ are both orthogonal. Let $C = A_2A_1^{-1}$, we have
\begin{eqnarray*}
C &=& A_2A_1^{-1} \\ 
&=& A_2\Sigma^{1/2} \Sigma^{-1/2}A_1^{-1}\\
&=& \left(A_2\Sigma^{1/2}\right) \left( A_1\Sigma^{1/2} \right)^{-1}
\end{eqnarray*}
Therefore, $C$ is orthogonal. And we have $\tilde\Theta_2 = C \tilde\Theta_1 C^{-1}$. Since $\tilde\Theta_1 + \tilde\Theta_1^T = D_1 \text{ and } \tilde\Theta_2 + \tilde\Theta_2^T = D_2$, we have
\begin{eqnarray*}
CD_1C^{-1} &=& C\left(\tilde\Theta_1 + \tilde\Theta_1^T\right)C^{-1} \\
&=& C\tilde\Theta_1C^{-1} + C\tilde\Theta_1^TC^{-1}\\
&=&  C\tilde\Theta_1C^{-1} + \left(C\tilde\Theta_1C^{-1}\right)^T\\
&=& \tilde\Theta_2 + \tilde\Theta_2^T\\
&=& D_2
\end{eqnarray*}
Now, suppose we have $D_1x = \lambda x$, then $D_2Cx = CD_1C^{-1} (Cx)
= CD_1x = \lambda Cx$. Therefore, $D_1$ and $D_2$ have the same
eigenvalues, which implies that $D_1 = D_2$.

Suppose $D_1 = D_2$ with $d_1, d_2, \ldots, d_m$ as the diagonal elements, where $d_1 \neq d_2\neq \ldots\neq d_m$. We know that 
$$D_1e_1 = d_1e_1$$
$$D_2Ce_1 = CD_1C^{-1}Ce_1 = CD_1e_1 = Cd_1e_1 = d_1 Ce_1$$ so $Ce_1
\propto e_1$, and since $C$ is orthogonal, this gives $Ce_1=\pm e_1$,
and similarly for $e_j$, $j = 1,2,\ldots,m$, we have $Ce_j =
\pm e_j$. Thus, $C$ is a diagonal matrix with all entries in $\{-1,1\}$. 
\end{proof}

\section{Proof of Theorem 2}
\label{AppB}

\begin{theorem}
  For two-dimensional space, if $\Theta$ and $\hat{\Theta}$ are sums
  of antisymmetric and diagonal matrices with decreasing diagonal
  elements, and with $\Theta_{12}\hat{\Theta}_{12}>0$, then
  $\mathlarger{\inf}_{\stackrel{\Theta_1 \in S_1}{\Theta_2\in S_2}}
  \left\{\left\lVert\Theta_1 -\Theta_2\right\rVert_F\right\} =
  \left\lVert\Theta -\hat\Theta\right\rVert_F$, where $S_1 = \{A\Theta
  A^{T}|AA^T =I\} $ and $S_2 = \{B\hat\Theta B^{T}|BB^T =I\} $.
\end{theorem}

\begin{proof}
Since the Frobenius norm is invariant under conjugation, for any
$\Theta_1 = A\Theta A^{T}$ and $\Theta_2=B\hat\Theta B^T$, where $A$
and $B$ are orthogonal, we have $
\left\lVert\Theta_1-\Theta_2\right\rVert_F=\left\lVert
A^T(\Theta_1-\Theta_2)A\right\rVert_F=\left\lVert\Theta-A^TB\hat\Theta
B^TA\right\rVert_F$, so it is sufficient to show that
$\mathlarger{\inf}_{\Theta_2\in S_2} \left\{\left\lVert\Theta
-\Theta_2\right\rVert_F\right\}=\left\lVert\Theta
-\hat\Theta\right\rVert_F$

In two dimensions, we have
$B = \begin{pmatrix} \cos\phi & \sin\phi \\ -\sin\phi & \cos\phi
\end{pmatrix}$
or
$B = \begin{pmatrix} \cos\phi & \sin\phi \\ \sin\phi & -\cos\phi
\end{pmatrix}$
for some $\phi$. We therefore want to show that if $\Theta$ and
$\hat\Theta$ are sums of antisymmetric and diagonal matrices with
decreasing diagonal elements, and
$\Theta_{12}\hat{\Theta}_{12}\geqslant 0$ then
$$\mathlarger{\inf}_{\phi} \left\{\left\lVert\Theta -\begin{pmatrix}
\cos\phi & \sin\phi \\
-\sin\phi  & \cos\phi 
\end{pmatrix}\hat\Theta\begin{pmatrix}
\cos\phi & -\sin\phi \\
\sin\phi  & \cos\phi 
\end{pmatrix}\right\rVert_F,
\left\lVert\Theta -\begin{pmatrix}
\cos\phi & \sin\phi \\
\sin\phi  & -\cos\phi 
\end{pmatrix}\hat\Theta\begin{pmatrix}
\cos\phi & \sin\phi \\
\sin\phi  & -\cos\phi 
\end{pmatrix}\right\rVert_F\right\}$$
$$=\left\lVert\Theta -\hat\Theta\right\rVert_F$$

Let $\Theta = \begin{pmatrix}
\theta_{11} & \theta_{12} \\
-\theta_{12}  & \theta_{22} 
\end{pmatrix}$ and $\hat\Theta = \begin{pmatrix}
\hat\theta_{11} & \hat\theta_{12} \\
-\hat\theta_{12}  & \hat\theta_{22} 
\end{pmatrix}$. Since $\left\lVert A\right\rVert_F^2 =
\operatorname{tr}(A^TA)$, we have 
\begin{align*}
\dfrac{d}{d\phi}\left\lVert\Theta -B \hat\Theta B^T\right\rVert_F^2 &= \dfrac{d}{d\phi}\operatorname{tr}\left\{\left(\Theta-B\hat\Theta B^T\right)^T\left(\Theta-B\hat\Theta B^T\right)\right\}\\
&= \dfrac{d}{d\phi}\operatorname{tr}\left\{\Theta^T\Theta - \Theta^T B\hat\Theta B^T - \left(\Theta^T B\hat\Theta B^T\right)^T + \left(B\hat\Theta B^T\right)^T\left(B\hat\Theta B^T\right)\right\}\\
&= -2\dfrac{d}{d\phi}\operatorname{tr}\left(\Theta^T B\hat\Theta B^T\right) \\
&= -2\dfrac{d}{d\phi}\sum_{i,j,k,l=1}^2 \Theta_{ji}\hat\Theta_{kl}B_{jk}B_{il} \\
\end{align*}

For
$B = \begin{pmatrix} \cos\phi & \sin\phi \\ -\sin\phi & \cos\phi
\end{pmatrix}$, we have 
$$\sum_{i,j,k,l=1}^2
\Theta_{ji}\hat\Theta_{kl}B_{jk}B_{il}=\left(\hat\theta_{11}\theta_{11}+\hat\theta_{22}\theta_{22}\right)\cos^2\phi
+ \left(\hat\theta_{22}\theta_{11}+\hat\theta_{11}\theta_{22}\right)\sin^2\phi +2\hat\theta_{12}\theta_{12}$$
while for
$B = \begin{pmatrix} \cos\phi & \sin\phi \\ \sin\phi & -\cos\phi
\end{pmatrix}$, we have 
$$\sum_{i,j,k,l=1}^2
\Theta_{ji}\hat\Theta_{kl}B_{jk}B_{il}=\left(\hat\theta_{11}\theta_{11}+\hat\theta_{22}\theta_{22}\right)\cos^2\phi+\left(\hat\theta_{22}\theta_{11}+\hat\theta_{11}\theta_{22}\right)\sin^2\phi -2\hat\theta_{12}\theta_{12}$$
In either case
\begin{align*}
\dfrac{d}{d\phi}\left\lVert\Theta -B \hat\Theta B^T\right\rVert_F^2 &=
-2\dfrac{d}{d\phi}\left(\left(\hat\theta_{11}\theta_{11}+\hat\theta_{22}\theta_{22}\right)\cos^2\phi+\left(\hat\theta_{22}\theta_{11}+\hat\theta_{11}\theta_{22}\right)\sin^2\phi\right) \\
&= 4\sin\phi\cos\phi\left(
\hat\theta_{22}\theta_{11}+\hat\theta_{11}\theta_{22}
-\hat\theta_{11}\theta_{11}-\hat\theta_{22}\theta_{22}
\right)
\end{align*}

To find the infimum, we set the derivative with respect to $\phi$ to
0. This gives us $\sin\phi = 0$ or $\cos\phi = 0$. This gives
$$B\in\left\{
\pm\begin{pmatrix}
1 & 0 \\
0  & 1
\end{pmatrix},
\pm\begin{pmatrix}
0 & 1 \\
-1  & 0 
\end{pmatrix},
\pm\begin{pmatrix}
1 & 0 \\
0  & -1
\end{pmatrix},
\right\}$$
This gives 
$$B\hat{\Theta}B^T\in\left\{
\begin{pmatrix}
\hat{\theta}_{11} & \hat{\theta}_{12} \\
-\hat{\theta}_{12}  & \hat{\theta}_{22}
\end{pmatrix},
\begin{pmatrix}
\hat{\theta}_{11} & -\hat{\theta}_{12} \\
\hat{\theta}_{12}  & \hat{\theta}_{22}
\end{pmatrix},
\begin{pmatrix}
\hat{\theta}_{22} & \hat{\theta}_{12} \\
-\hat{\theta}_{12}  & \hat{\theta}_{11}
\end{pmatrix},
\begin{pmatrix}
\hat{\theta}_{22} & -\hat{\theta}_{12} \\
\hat{\theta}_{12}  & \hat{\theta}_{11}
\end{pmatrix}
\right\}$$
For these matrices $B$, the corresponding squared Frobenius norms
$\lVert \Theta-B^T\hat{\Theta}B\rVert^2$ are
\begin{align*}
(\theta_{11}-\hat{\theta}_{11})^2+(\theta_{22}-\hat{\theta}_{22})^2+2(\theta_{12}-\hat{\theta}_{12})^2\\
(\theta_{11}-\hat{\theta}_{11})^2+(\theta_{22}-\hat{\theta}_{22})^2+2(\theta_{12}+\hat{\theta}_{12})^2\\
(\theta_{11}-\hat{\theta}_{22})^2+(\theta_{22}-\hat{\theta}_{11})^2+2(\theta_{12}-\hat{\theta}_{12})^2\\
  (\theta_{11}-\hat{\theta}_{22})^2+(\theta_{22}-\hat{\theta}_{11})^2+2(\theta_{12}+\hat{\theta}_{12})^2
\end{align*}
Since $\theta_{11}>\theta_{22}$, $\hat{\theta}_{11}>\hat{\theta}_{22}$
and $\theta_{12}\hat{\theta}_{12}>0$, it is easy to see that the first
of these is the smallest.
\end{proof}

\newpage
\section{Simulation Setup}
\label{AppC}
\begin{center}

\begin{table}[h]
\scriptsize
\centering
\begin{tabular}{m{4cm} m{8cm}}
\hline\hline
\multicolumn{2}{c} {Z Matrix} \\\hline\hline
 Value & Note  \\\hline
 \hfill\break
$\begin{bmatrix}
0.7 & 0.8 \\
0.1 & -0.1 
\end{bmatrix}$ \hfill\break & Two dimensional output equation \hfill \break Two vectors have small angle\\\hline
 \hfill\break
$\begin{bmatrix}
0.1 & 0.2 \\
0.8 & -0.1\\
4 & 1.5  
\end{bmatrix}$ \hfill\break & Three dimensional output equation \hfill \break Two vectors have small angle\\\hline
 \hfill\break
$\begin{bmatrix}
0.1 & 0.2 \\
0.8 & -0.1\\
4 & 1.5 \\
-0.5 & -0.2 
\end{bmatrix}$ \hfill\break & Four dimensional output equation \hfill \break Two vectors have small angle\\\hline
 \hfill\break
$\begin{bmatrix}
0.2 & 0.5 \\
0.5 & -0.2 
\end{bmatrix}$ \hfill\break & Two dimensional output equation \hfill \break Two vectors are orthogonal\\\hline
 \hfill\break
$\begin{bmatrix}
0.1 & 0.2 \\
0.8 & -0.1\\
-0.2 & -0.3 
\end{bmatrix}$ \hfill\break & Three dimensional output equation \hfill \break Two vectors are orthogonal\\\hline
 \hfill\break
$\begin{bmatrix}
0.1 & 0.2 \\
0.8 & -0.1\\
-0.2 & -0.9 \\
0.6 & -0.2
\end{bmatrix}$ \hfill\break & Four dimensional output equation \hfill \break Two vectors are orthogonal\\\hline
\end{tabular}
\caption{Six different simulation scenarios of $Z$.}
\label{Zsim}
\end{table}
\end{center}

\begin{center}
\begin{table}[H]
\scriptsize
\centering
\begin{tabular}{m{4cm} m{8cm}}
\hline\hline
\multicolumn{2}{c} {$\Theta$ Matrix} \\\hline\hline
 Value & Note  \\\hline
$\begin{pmatrix}
0.16 & 0.02 \\
-0.02 & 0.1 
\end{pmatrix}$ & Two real roots \hfill \break Diagonal elements are close and small \hfill \break
$(\lambda_1 - \lambda_2)^2$ small
\\\hline
$\begin{bmatrix}
0.8 & 0.3 \\
-0.3 & 0.1 
\end{bmatrix}$ & Two real roots \hfill \break Diagonal elements are far and small\hfill \break
$(\lambda_1 - \lambda_2)^2$ small\\\hline
$\begin{bmatrix}
2.4 & 0.1 \\
-0.1 & 2 
\end{bmatrix}$ & Two real roots \hfill \break Diagonal elements are close and large\hfill \break
$(\lambda_1 - \lambda_2)^2$ small\\\hline
$\begin{bmatrix}
2.4 & 0.5 \\
-0.5 & 1 
\end{bmatrix}$ & Two real roots \hfill \break Diagonal elements are far and large\hfill \break
$(\lambda_1 - \lambda_2)^2$ small\\\hline
$\begin{bmatrix}
0.16 & 0.04 \\
-0.04 & 0.1 
\end{bmatrix}$ & Two complex roots, small imaginary part \hfill \break Diagonal elements are close and small\hfill \break
$(\lambda_1 - \lambda_2)^2$ small\\\hline
$\begin{bmatrix}
0.8 & 0.4 \\
-0.4 & 0.1 
\end{bmatrix}$ & Two complex roots, small imaginary part \hfill \break Diagonal elements are far and small\hfill \break
$(\lambda_1 - \lambda_2)^2$ small\\\hline
$\begin{bmatrix}
2.4 & 0.25 \\
-0.25 & 2 
\end{bmatrix}$ & Two complex roots, small imaginary part \hfill \break Diagonal elements are close and large\hfill \break
$(\lambda_1 - \lambda_2)^2$ small\\\hline
$\begin{bmatrix}
2.4 & 0.75 \\
-0.75 & 1 
\end{bmatrix}$ & Two complex roots, small imaginary part \hfill \break Diagonal elements are far and large\hfill \break
$(\lambda_1 - \lambda_2)^2$ small\\\hline
$\begin{bmatrix}
0.16 & 1 \\
-1 & 0.1 
\end{bmatrix}$ & Two complex roots, large imaginary part \hfill \break Diagonal elements are close and small\hfill \break
$(\lambda_1 - \lambda_2)^2$ large\\\hline
\end{tabular}
\caption{Fifteen different simulation scenarios of $\Theta$ (Part 1).}
\label{Thetasim}
\end{table}
\end{center}

\begin{center}
\begin{table}[h]
\scriptsize
\centering
\begin{tabular}{m{4cm} m{8cm}}
\hline\hline
\multicolumn{2}{c} {$\Theta$ Matrix} \\\hline\hline
 Value & Note  \\\hline
 $\begin{bmatrix}
0.8 & 1.8 \\
-1.8 & 0.1 
\end{bmatrix}$ & Two complex roots, large imaginary part \hfill \break Diagonal elements are far and small\hfill \break
$(\lambda_1 - \lambda_2)^2$ large\\\hline
$\begin{bmatrix}
2.4 & 2 \\
-2 & 2 
\end{bmatrix}$ & Two complex roots, large imaginary part \hfill \break Diagonal elements are close and large\hfill \break
$(\lambda_1 - \lambda_2)^2$ large\\\hline
$\begin{bmatrix}
2.4 & 2 \\
-2 & 1 
\end{bmatrix}$ & Two complex roots, large imaginary part \hfill \break Diagonal elements are far and large\hfill \break
$(\lambda_1 - \lambda_2)^2$ large\\\hline
$\begin{bmatrix}
6 & 2 \\
-2 & 1.8 
\end{bmatrix}$ & Two real roots, large imaginary part \hfill \break Diagonal elements are far and large\hfill \break
$(\lambda_1 - \lambda_2)^2$ large\\\hline
$\begin{bmatrix}
6 & 1.5 \\
-1.5 & 1.8 
\end{bmatrix}$ & Two real roots, large imaginary part \hfill \break Diagonal elements are far and large\hfill \break
$(\lambda_1 - \lambda_2)^2$ large\\\hline
$\begin{bmatrix}
6 & 0.5 \\
-0.5 & 1.8 
\end{bmatrix}$ & Two real roots, large imaginary part \hfill \break Diagonal elements are far and large\hfill \break
$(\lambda_1 - \lambda_2)^2$ large\\\hline
\end{tabular}
\caption{Fifteen different simulation scenarios of $\Theta$ (Part 2).}
\label{Thetasim2}
\end{table}
\end{center}

\newpage
\section{BIC of Simulated Data}
\label{AppD}
\begin{figure}[h]
\centering
\includegraphics[width=\textwidth]{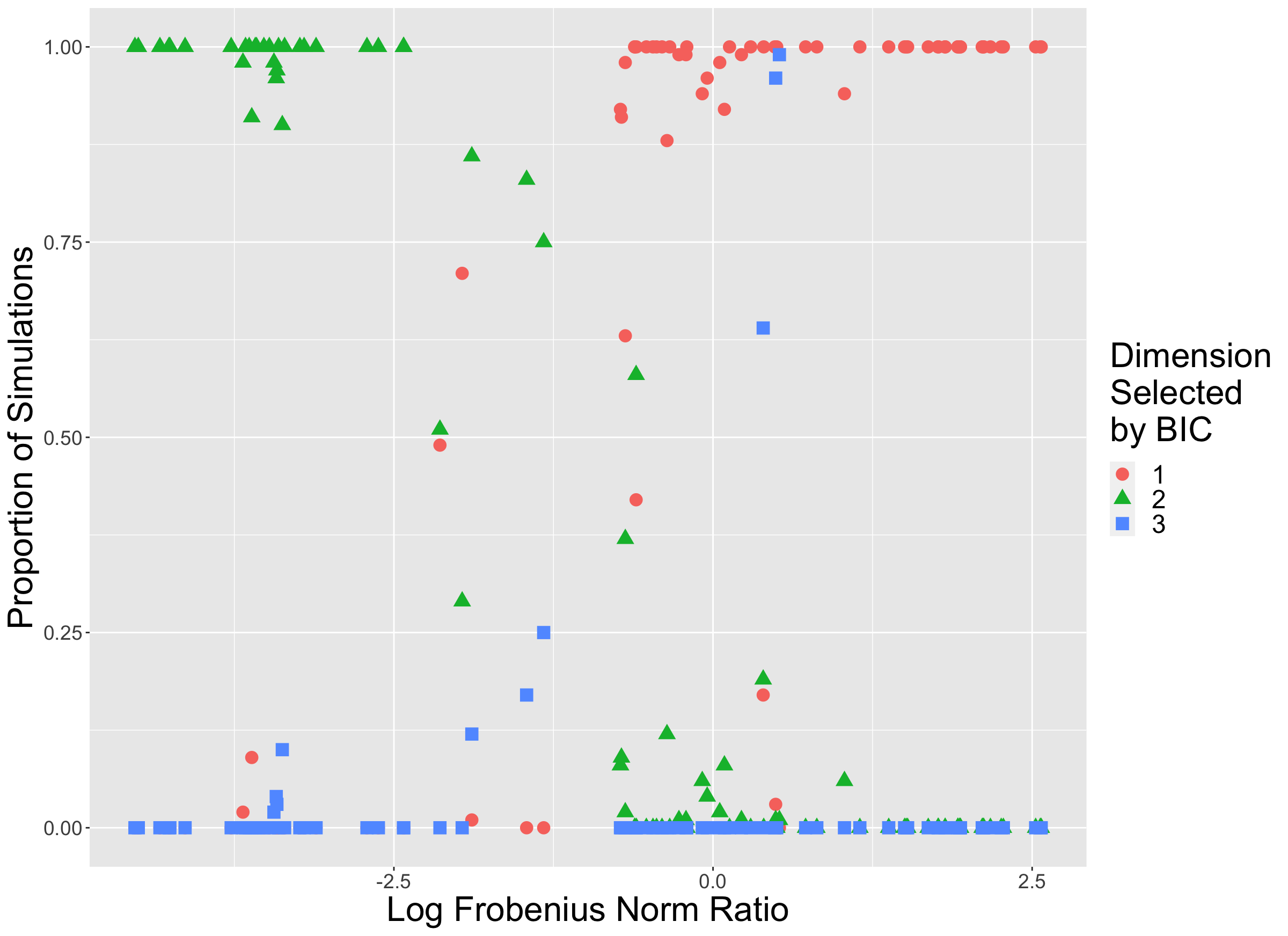}
\caption{Proportion of times BIC selects different dimensions for the state equation versus the average Logarithm of Frobenius norm ratio between $\exp{(-\hat{\Theta})}$ and true $\exp{(-\Theta})$ value over 100 simulations for each simulation scenario. }
\label{BIC}
\end{figure}

\end{document}